\pdfoutput=1
\documentclass[runningheads]{llncs}

\usepackage{cite}
\usepackage{amsmath,amssymb,amsfonts}
\usepackage{graphicx}
\usepackage{textcomp}
\usepackage[dvipsnames]{xcolor}
\usepackage[font=small,skip=1pt]{caption}
\usepackage[caption=false,font=footnotesize]{subfig}
\usepackage{enumitem}
\usepackage[export]{adjustbox}
\usepackage{multirow}
\usepackage{float}
\usepackage{bigdelim}
\usepackage{makecell}
\usepackage{listings}
\usepackage[plain]{algorithm}
\usepackage[noend]{algpseudocode}
\usepackage{url}

\usepackage[breaklinks,colorlinks=true]{hyperref}
\usepackage{tlatex}
\usepackage{stmaryrd}
\usepackage{centernot}
\usepackage{comment}
\usepackage{empheq}
\usepackage{appendix}
\usepackage{empheq}
\usepackage{amssymb}
\usepackage{pifont}
\usepackage[normalem]{ulem}
\usepackage{setspace}
\usepackage{scalerel}
\usepackage{bm}
\usepackage{xspace}
\usepackage{soul}

\usepackage{tikz}
\usetikzlibrary{shapes.geometric, arrows}

\renewenvironment{comment}{}{}

\newcommand{\mathcolorbox}[2]{\colorbox{#1}{$\displaystyle #2$}}

\newcommand{\cmark}{\ding{51}}%
\newcommand{\xmark}{\ding{55}}%

\setlist[itemize]{leftmargin=*}
\setlist[itemize]{noitemsep, topsep=0pt}

\makeatletter
\def\@citecolor{blue}%
\def\@urlcolor{blue}%
\def\@linkcolor{blue}%

\def\orcidID#1{\smash{\href{http://orcid.org/#1}{\protect\raisebox{-1.25pt}{\protect\includegraphics{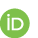}}}}}
\makeatother

\usepackage{bbding}

\newcommand{\tlasize}{\fontsize{8pt}{11pt}\selectfont}
\newcommand{\eqsize}{\small}
\newcommand{\eqsizesmall}{\small}

\definecolor{nonecolor}{HTML}{666666}
\definecolor{highcolor}{HTML}{790C0C}
\definecolor{failcolor}{HTML}{666666}

\definecolor{hlcolor}{HTML}{DCDCDC}
\DeclareRobustCommand{\hlc}[1]{{\sethlcolor{hlcolor}\hl{#1}}}

\definecolor{commentcolor}{HTML}{787c82}

\algnewcommand\algorithmicforeach{\textbf{for each}}
\algdef{S}[FOR]{ForEach}[1]{\algorithmicforeach\ #1\ \algorithmicdo}

\renewcommand{\implies}{\Rightarrow}

\newcommand{\fin}[1]{{\tt #1}}

\newcommand{\lo}[1]{{\fin{#1}}}

\newcommand{\blind}[1]{}
\newcommand{\hide}[1]{}

\usepackage{mathtools}

\newcommand{\regularf}[1]{\textit{#1}}

\algnewcommand{\algorithmicgoto}{\textbf{go to} Line}%
\algnewcommand{\Goto}[1]{\algorithmicgoto~\ref{#1}}%

\algnewcommand{\algorithmicnot}{\textbf{not}}%
\algnewcommand{\Not}{\algorithmicnot}%

\newcommand{\quant}{\Phi}
\newcommand{\icpo}{IC3PO\xspace}
\newcommand{\updr}{UPDR\xspace}
\newcommand{\folic}{fol-ic3\xspace}
\newcommand{\ifour}{I4\xspace}

\newcommand{\phih}{\widehat{\varphi}}
\newcommand{\s}{\fin{s}}

\newcommand{\same}{\equiv}

\newcommand{\sz}{\sigma}
\newcommand{\szk}{\sz}
\newcommand{\sznext}{\szk^{\texttt{+}}[\s_{\fin{i}}]}
\newcommand{\szinit}{\text{\textit{base size}}}

\newcommand{\invk}{Inv}
\newcommand{\cexk}{Cex}

\newcommand{\pro}{\mathcal{P}}
\newcommand{\fpro}{\hat{\pro}}
\newcommand{\fprok}{\fpro}

\newcommand{\fpronext}{\pro(\sznext)}
\newcommand{\numc}{\texttt{\#}}

\newcommand{\naiveconsensus}{\textit{ToyConsensus}\xspace}
\newcommand{\naiveinstance}{\textit{ToyConsensus}(3,3,3)}

\newcommand{\defsym}{$\triangleq$\xspace}
\newcommand{\ndefsym}{\phantom{$\triangleq$}\xspace}
\newcommand{\altsym}{\AE\xspace}
\newcommand{\eprsym}{\phantom{$\leftrightarrows$}\xspace}
\newcommand{\neprsym}{$\leftrightarrows$\xspace}

\newcommand{\symic}{\textit{SymIC3}\xspace}
\newcommand{\symboost}{\textit{SymBoost$\forall\exists$}\xspace}

\begin{document}

\title{On Symmetry and Quantification: \\ A New Approach to Verify Distributed Protocols}

\titlerunning{On Symmetry and Quantification}
%
\author{Aman Goel\inst{(}\Envelope\inst{)}\orcidID{0000-0003-0520-8890} \and Karem Sakallah}
\authorrunning{A. Goel and K. Sakallah}
%
\institute{University of Michigan, Ann Arbor MI 48105, USA\\
\email{\{amangoel,karem\}@umich.edu}}
\maketitle              

\begin{abstract}
Proving that an unbounded distributed protocol satisfies a given safety property amounts to finding a quantified inductive invariant that implies the property for all possible instance sizes of the protocol.
Existing methods for solving this problem can be described
as search procedures for an invariant whose quantification prefix fits a particular template. 
We propose an alternative
\textit{constructive} approach that does not prescribe, \textit{a priori}, a specific quantifier prefix. Instead, the required prefix is automatically
inferred without any search by carefully analyzing the structural symmetries
of the protocol.
The key insight underlying this approach is that symmetry and quantification are closely related concepts that express protocol  \textit{invariance} under different re-arrangements of its components. 
We propose \textit{symmetric incremental induction}, an extension of the finite-domain IC3/PDR algorithm, that automatically derives the required \textit{quantified inductive invariant} by exploiting the connection between symmetry and quantification. 
While various attempts have been made to exploit symmetry in verification applications, to our knowledge, this is the first demonstration of a direct link between symmetry and quantification in the context of clause learning during incremental induction.
We also describe a procedure to automatically find a minimal finite size, the \textit{cutoff}, that yields a quantified invariant proving safety for any size.

\hspace{12pt} Our approach is implemented in \icpo, a new verifier for distributed protocols that significantly outperforms the state-of-the-art, scales orders of magnitude faster, and robustly derives compact inductive invariants fully automatically.

\end{abstract}

\section{Introduction}
\label{sec:introduction}
Our focus in this paper is on \textit{parameterized verification}, specifically proving \textit{safety} properties of distributed systems, such as protocols that are often modeled above the code level (e.g.,~\cite{lamport2002specifying,padon2016ivy}), consisting of arbitrary numbers of \textit{identical} components that are instances of a small set of different \textit{sorts}. For example, a client server protocol\cite{ivyclientserver} $CS(i,j)$ is a two-sort parameterized system with parameters $i \geq 1$ and $j \geq 1$ denoting, respectively, the number of clients and servers. 
Protocol correctness proofs are critical for establishing the correctness of actual system implementations in established methodologies such as~\cite{hawblitzel2015ironfleet,Wilcox2015Verdi}.
Proving safety properties for such systems requires the derivation of inductive invariants that are expressed as state predicates quantified over the system parameters.   
While, in general, this problem is undecidable~\cite{apt1986limits}, certain restricted forms have been shown to yield to algorithmic solutions~\cite{bloem2015decidability}. Key to these solutions is appealing to the problem's inherent symmetry. In this paper, we exclusively focus on protocols whose sorts represent sets of indistinguishable domain constants. The behavior of this restricted class of protocols remains invariant under all possible permutations of the domain constants. We leave the exploration of other features, such as totally-ordered sorts, integer arithmetic, etc., for future work.

Our proposed  symmetry-based solution is best understood by briefly reviewing earlier efforts. Initially, the pressing issue was the inevitable \textit{state explosion} when verifying a finite, but large, parameterized system~\cite{emerson1996symmetry, ip1996better, pong1995new, godefroid1999exploiting, sistla2000smc,  barner2002combining}. Thus, instead of verifying the ``full'' system, these approaches verified its \textit{symmetry-reduced quotient,} mostly using BDD-based symbolic image computation~\cite{burch1990symbolic,burch1992symbolic,mcmillan1993symbolic}. The Mur$\varphi$ verifier~\cite{ip1996better} was a notable exception in that it a)  generated a C++ program that enumerated the system's symmetry-reduced reachable states, and b) allowed for the verification of unbounded systems by taking advantage of \textit{data saturation} which happens when the size of the symmetry-reduced reachable states become constant regardless of system size. 

The idea that an unbounded \textit{symmetric} system can, under certain data-independence assumptions,  be verified by analyzing small finite instances evolved into the approach of verification by \textit{invisible invariants}~\cite{pnueli2001automatic, arons2001parameterized, zuck2004model,balaban2005iiv,dooley2016proving}. In this approach, assuming they exist, inductive invariants that are universally-quantified over the system parameters are automatically derived by analyzing instances of the system up to a \textit{cutoff} size $N_0$ using a combination of symbolic reachability and symmetry-based abstraction.
Noting that an invariant is an over-approximation of the reachable states, the restriction to universal quantification may fail in some cases, rendering the approach incomplete. The invisible invariant verifier IIV~\cite{balaban2005iiv} employs some heuristics to derive invariants that use combinations of universal and existential  quantifiers, but as pointed out in \cite{namjoshi2007symmetry}, it may still fail and is not guaranteed to be complete. 

The development of SAT-based incremental induction algorithms~\cite{bradley2011sat,een2011efficient} for verifying the safety of finite transition systems was a major advance in the field of model checking and has, for the most part, replaced BDD-based approaches. These algorithms leverage the capacity and performance of modern CDCL SAT solvers~\cite{marques1999grasp,moskewicz2001chaff,een2003extensible,balyo2020proceedings} to  produce \textit{clausal strengthening assertions} $A$ that, conjoined with a specified safety property $P$, form  an automatically-generated inductive invariant $Inv = A \wedge P$ if the property holds. The AVR hardware verifier~\cite{goel2020avr,goel2019model,goel2019empirical} was adapted in~\cite{ma2019i4} to produce quantifier-free inductive invariants for small instances of unbounded protocols that are subsequently generalized with universal quantification, in analogy with the invisible invariants approach,  to arbitrary sizes. The resulting assertions tended, in some cases, to be quite large, and the approach was also incomplete due to the restriction to universal quantification.

In this paper we introduce \icpo, a novel symmetry-based verifier that builds on these previous efforts while removing most of their limitations. Rather than search for an invariant with a prescribed quantifier prefix, \icpo constructively \textit{discovers} the required quantified assertions by performing \textit{symmetric incremental induction} and analyzing the symmetry patterns in learned clauses to infer the corresponding quantifier prefix.
Our main contributions are:
\begin{itemize}
    \item An extension to finite incremental induction algorithms that uses protocol symmetry to boost clause learning from a \textit{single} clause $\varphi$ to a set of symmetrically-equivalent clauses, $\varphi$'s \textit{orbit}. 
    \item A quantifier inference procedure that expresses $\varphi$'s orbit by an automatically-derived \textit{compact} quantified predicate $\Phi$. The inference procedure is based on a simple analysis of $\varphi$'s \textit{syntactic structure} and yields a quantified form with both universal and existential quantifiers.    
    \item A systematic \textit{finite convergence} procedure for determining a minimal instance size sufficient for deriving a quantified inductive invariant that holds for all sizes.
\end{itemize}
We also demonstrate the effectiveness of \icpo~on a diverse set of benchmarks and show that it significantly advances the current state-of-the-art.

The paper is structured as follows:
\S\ref{sec:notation} presents preliminaries. \S\ref{sec:protocols} formalizes protocol symmetries.
The next three sections detail our key contributions: symmetry boosting during incremental induction in \S\ref{sec:symclause}, relating symmetry to quantification in \S\ref{sec:quantinfer}, and checking for convergence in \S\ref{sec:invcheck}.
\S\ref{sec:ic3po_highlevel} describes the \icpo algorithm and implementation details. \S\ref{sec:evaluation} presents our experimental evaluation. The paper concludes with a brief survey of related work in \S\ref{sec:related}, and a discussion of future directions in \S\ref{sec:conclusions}.

\section{Preliminaries}
\label{sec:notation}
\noindent Figure~\ref{fig:tla_naive} describes a toy consensus protocol from~\cite{toyconsensus} in the TLA+ language~\cite{lamport2002specifying}.\footnote{The description in~\cite{toyconsensus} is in the Ivy~\cite{padon2016ivy} language and encodes set operations in relational form with a $member$ relation representing $\in$.} The protocol has three named sorts $S = [ \fin{node}, \fin{quorum}, \fin{value} ]$ introduced by the \textsc{constants} declaration, and two relations $R = \{vote, decision\}$, introduced by the \textsc{variables} declaration, that are defined on these sorts. Each of the sorts is understood to represent an unbounded domain of distinct elements with the relations serving as the protocol's state variables. The global axiom (line 3) defines the elements of the \fin{quorum} sort to be subsets of the \fin{node} sort and restricts them further by requiring them to be pair-wise non-disjoint. We will refer to \fin{node} (resp. \fin{quorum}) as an \textit{independent} (resp. \textit{dependent}) sort. The protocol transitions are specified by the actions $CastVote$ and $Decide$ (lines 6-7) which are expressed using the current- and next-state  variables as well as the definitions $didNotVote$ and $chosenAt$ (lines 4-5) which serve as \textit{auxiliary non-state} variables.
Lines 8-10 specify the protocol's initial states, transition relation, and safety property.
\begin{algorithm}[!tb]
\batchmode 
\tlatex
\setstretch{0.85}
\tlasize
\setboolean{shading}{true}
\@x{\makebox[5pt][r]{\scriptsize \hspace{1em}}}\moduleLeftDash\@xx{ {\MODULE} \text{\naiveconsensus}}\moduleRightDash\@xx{\makebox[5pt][r]{\scriptsize \hspace{1em}}}%
\@x{\makebox[25pt][r]{\scriptsize 1\hspace{1.3em}} {\CONSTANTS} \fin{node} ,\, \fin{quorum} ,\, \fin{value} \@s{60} {\VARIABLES} vote ,\, decision}%
\@pvspace{5.0pt}%
 \@x{\makebox[25pt][r]{\scriptsize 2\hspace{1.3em}} vote\@s{2} \.{\in} ( \fin{node} \times \fin{value}) \.{\rightarrow}
 {\BOOLEAN} \@s{47} decision\@s{2} \.{\in} \fin{value} \.{\rightarrow} {\BOOLEAN} }%
\@pvspace{5.0pt}%
 \@x{\makebox[25pt][r]{\scriptsize 3\hspace{1.3em}} {\ASSUME}\@s{7} \A\, Q \.{\in} \fin{quorum} \.{:} Q \subseteq \fin{node} \@s{4} \.{\land}  \@s{4} \A\, Q_1 ,\, Q_2 \.{\in} \fin{quorum} \.{:} Q_1 \.{\cap} Q_2 \.{\neq} \{ \} } %
\@pvspace{5.0pt}%
\@x{\makebox[25pt][r]{\scriptsize 4\hspace{1.3em}} didNotVote( n ) \@s{4} \.{\defeq} \.\A\, V \.{\in} \fin{value} \.{:} {\neg} vote ( n ,\, V )}
\@pvspace{5.0pt}%
\@x{\makebox[25pt][r]{\scriptsize 5\hspace{1.3em}} chosenAt( q ,\, v ) \@s{3} \.{\defeq} \.\A\, N \.{\in} q \.{:} vote ( N ,\, v )} %
\@pvspace{5.0pt}%
 \@x{\makebox[25pt][r]{\scriptsize 6\hspace{1.3em}} CastVote ( n ,\, v ) \@s{2} \.{\defeq} \@s{2} \.didNotVote( n ) \@s{4} \.{\land} \@s{4} \.vote \.{'} \.{=} [ vote {\EXCEPT} {\bang} [ n ,\, v
 ] \.{=} {\TRUE} ]}%
\@x{\makebox[25pt][r]{\scriptsize \hspace{1.3em}} \@s{136} \.{\land} \@s{4} \.{\UNCHANGED} decision}%
\@pvspace{5.0pt}%
 \@x{\makebox[25pt][r]{\scriptsize 7\hspace{1.3em}} Decide ( q ,\, v ) \@s{12} \.{\defeq} \@s{2} \.chosenAt( q ,\, v ) \@s{4} \.{\land} \@s{4} \.decision \.{'} \.{=} [ decision {\EXCEPT} {\bang} [
 v ] \.{=} {\TRUE} ]}%
 \@x{\makebox[25pt][r]{\scriptsize \hspace{1.3em}}   \@s{137.5}\.{\land} \@s{4} \.{\UNCHANGED} vote}%
\@pvspace{5.0pt}%
\@x{\makebox[25pt][r]{\scriptsize 8\hspace{1.3em}} Init \@s{2} \.{\defeq} \@s{2} \.\A\, N\@s{0.19} \.{\in} \fin{node} ,\, V \.{\in} \fin{value} \.{:} {\neg} vote ( N ,\, V ) \@s{4} \.{\land} \@s{4} \.\A\, V \.{\in} \fin{value} \.{:} {\neg} decision ( V )} %
\@pvspace{5.0pt}%
\@x{\makebox[25pt][r]{\scriptsize 9\hspace{1.3em}} T \@s{9} \.{\defeq} \@s{2} \.\E\, N {\in} \fin{node} , Q {\in} \fin{quorum}, V {\in} \fin{value} \,{:}\, CastVote(N ,\, V) \.{\lor} Decide(Q ,\, V)}%
\@pvspace{5.0pt}%
 \@x{\makebox[25pt][r]{\scriptsize 10\hspace{1.3em}} P \@s{9} \.{\defeq} \@s{2} \.\A\, V_1 ,\, V_2 \.{\in} \fin{value} \.{:} decision ( V_1 )
 \.{\land} decision ( V_2 ) \.{\implies} V_1 \.{=} V_2}%
\@x{\makebox[5pt][r]{\scriptsize \hspace{1em}}}\bottombar\@xx{\makebox[5pt][r]{\scriptsize \hspace{1em}}}%
\captionsetup{belowskip=-10pt}
\captionof{figure}{Toy consensus protocol in the TLA+ language}
\label{fig:tla_naive}
\end{algorithm}%

Viewed as a parameterized system, the \textit{template} of an arbitrary $\fin{n}$-sort distributed protocol $\pro$ will be expressed as $\pro(\fin{s_1,\dots,s_n})$ where $S = [\fin{s_1,\dots,s_n}]$ is an ordered list of its sorts, each of which is assumed to be an unbounded uninterpreted set of distinct \textit{constants}. As a mathematical transition system, $\pro$ is defined by a) its state variables which are expressed as $k$-ary relations on its sorts, and b) its actions which capture its state transitions. We also note that non-Boolean functions/variables can be easily accommodated by encoding them in relational form, e.g.,  $f(\fin{x_1}, \fin{x_2}, \dots) = \fin{y}$. We will use $Init, T,$ and $P$ to denote, respectively, a protocol's initial states, its transition relation, and a safety property that is required to hold on all reachable states.  A finite instance of $\pro$ will be denoted as $\pro(\fin{|s_1|,\dots,|s_n|})$ where each named sort is replaced by its finite size in the instance. Similarly, $Init(\fin{|s_1|,\dots,|s_n|})$, $T(\fin{|s_1|,\dots,|s_n|})$ and $P(\fin{|s_1|,\dots,|s_n|})$ will, respectively, denote the application of $Init$, $T$ and $P$ to this finite instance.

The template of the protocol in Figure~\ref{fig:tla_naive} is $ToyConsensus(\fin{node,quorum,value})$. Its finite instance:
{
\eqsize
\begin{align}
\naiveinstance : \hspace{10pt} &\fin{node_3} \triangleq \{ \fin{n_1}, \fin{n_2}, \fin{n_3} \}  \hspace{30pt} \fin{value_3} \triangleq \{ \fin{v_1, v_2, v_3} \} \label{eq:ex_size}\\
& \fin{quorum_3} \triangleq \{ \fin{q_{12}\!:\!\{n_1,n_2\},~q_{13}\!:\!\{n_1,n_3\},~q_{23}\!:\!\{n_2,n_3\}} \} \nonumber
\end{align}
}%
will be used as a running example in the paper.
The finite sorts of this instance are defined as sets of arbitrarily-named distinct constants. It should be noted that the constants of the $\fin{quorum_3}$ sort are subsets of the $\fin{node_3}$ sort that satisfy the non-empty intersection axiom and are named to reflect their symmetric dependence on the $\fin{node_3}$ sort.
This instance has 9 $vote$ and 3 $decision$ state variables, and a $state$ of this instance corresponds to a complete Boolean assignment to these 12 state variables. 

In the sequel, we will use $\fpro$ and $\hat{T}$ as shorthand for $\pro(\fin{|s_1|, \dots, |s_n|})$ and  $T(\fin{|s_1|,\dots,|s_n|})$.
Quantifier-free formulas will be denoted by lower-case Greek letters (e.g., $\varphi$) and quantified formulas by upper-case Greek letters (e.g., $\Phi$).
We use primes (e.g., $\varphi'$) to represent a formula after a single transition step.

\section{Protocol Symmetries}
\label{sec:protocols}
\noindent
The symmetry group of $\fpro$ is $G(\fpro) = \bigtimes_{~\s \in S} Sym(\s)$,
where $Sym(\s)$ is the symmetric group, i.e., the set of $|\s|!$ permutations of the constants of the set $\s$.\footnote{We assume familiarity with basic notions from \textit{group theory} including \textit{permutation groups}, \textit{cycle notation}, \textit{group action} on a set, \textit{orbits}, etc.,  which can be readily found in standard textbooks on Abstract Algebra~\cite{Fraleigh00}.} 
In what follows we will use $G$ instead of $G(\fpro)$ to reduce clutter.
Given a permutation $\gamma \in G$ and an arbitrary protocol relation $\rho$ instantiated with specific sort constants, the \emph{action} of $\gamma$ on $\rho$, denoted $\rho^\gamma$, is the relation  obtained from $\rho$ by permuting the sort constants in $\rho$ according to $\gamma$; it is referred to as the $\gamma$-\emph{image} of $\rho$. Permutation $\gamma \in G$ can also act on any formula involving the protocol relations. In particular, the invariance of protocol behavior under permutation of sort constants implies that the action of $\gamma$ on the (finite) initial state, transition relation, and property formulas causes a syntactic re-arrangement of their sub-formulas while preserving their logical equivalence:
{
\eqsize
\begin{equation}
\hfill
\hat{Init}^{\gamma} \same \hat{Init}
\qquad \qquad \qquad
\hat{T}^{\gamma} \same \hat{T}
\qquad \qquad \qquad
\hat{P}^{\gamma} \same \hat{P}
\hfill
\label{eqn:ITP_invariance}
\end{equation}}%

Consider next a clause $\varphi$ which is a disjunction of literals, namely, instantiated protocol relations or their negations.
The \emph{orbit} of $\varphi$ under $G$, denoted $\varphi^G$, is the set of its images $\varphi^\gamma$ for all permutations $\gamma \in G$, i.e., ${\varphi ^G} = \left\{ {\left. {{\varphi ^\gamma }} \right|\gamma  \in G} \right\}$.
The $\gamma$-image of a clause can be viewed as a \emph{syntactic} transformation that will either yield a new logically-distinct clause on different literals or simply re-arrange the literals in the clause without changing its logical behavior (by the commutativity and associativity of disjunction). We define the \emph{logical action} of a permutation $\gamma$ on a clause $\varphi$, denoted $\varphi^{L(\gamma)}$, as:
{
\vspace{-1pt}
\eqsize
$${\varphi ^{L(\gamma )}} = \left\{ {\begin{array}{*{20}{l}}
  {{\varphi ^\gamma }}&{{\text{if }}{\varphi ^\gamma } \not \same \varphi } \\ 
  \varphi &{{\text{if }}{\varphi ^\gamma } \same \varphi } 
\end{array}} \right.$$
}%
and the \emph{logical orbit} of $\varphi$ as ${\varphi ^{L(G)}} = \left\{ {\left. {\varphi ^{L(\gamma )}} \right|{\gamma  \in G}} \right\}$.
With a slight abuse of notation, logical orbit can also be viewed as the conjunction of the logical images:
{
\vspace{-1pt}
\eqsize
\[{\varphi ^{L(G)}} = \mathop  \bigwedge \limits_{\gamma  \in G} {\varphi ^{L(\gamma )}}\]
}%

To illustrate these concepts, consider $\naiveinstance$ from (\ref{eq:ex_size}).
Its symmetries in cycle notation are as follows:
{
\eqsizesmall
\begin{align}
& \begin{aligned}
& Sym(\fin{node_3}) &=~ \{ (), (\fin{n_1~n_2}), (\fin{n_1~n_3}), (\fin{n_2~n_3}), (\fin{n_1~n_2~n_3}), (\fin{n_1~n_3~n_2}) \} \\
& Sym(\fin{value_3}) &=~ \{ (), (\fin{v_1~v_2}), (\fin{v_1~v_3}), (\fin{v_2~v_3}), (\fin{v_1~v_2~v_3}), (\fin{v_1~v_3~v_2}) \}
\end{aligned}  \nonumber\\
& G =~ Sym(\fin{node_3}) \times Sym(\fin{value_3}) \label{eq:ex_sym_group}
\end{align}}%
The symmetry group (\ref{eq:ex_sym_group}) of $\naiveinstance$ has $36$ symmetries corresponding to the $6$ $\fin{node_3}$ $\times$ $6$ $\fin{value_3}$ permutations. The permutations on $\fin{quorum_3}$ are \textit{implicit} and based on the permutations of $\fin{node_3}$ since $\fin{quorum_3}$ is a dependent sort. Now, consider the example clause:
{
\vspace{-1pt}
\eqsize
\begin{align}
\varphi_1 =& ~vote(\fin{n_1, v_1}) \vee vote(\fin{n_1, v_2}) \vee vote(\fin{n_1, v_3}) \label{eq:ex_clause1}
\end{align}}%
The orbit of $\varphi_1$ consists of 36 syntactically-permuted clauses. However, many of these images are logically equivalent yielding the following logical orbit of just $3$ logically-distinct clauses:
{
\vspace{-1pt}
\eqsizesmall
\begin{align}
\varphi_1^{L(G)} =& ~[~ vote(\fin{n_1, v_1}) \vee vote(\fin{n_1, v_2}) \vee vote(\fin{n_1, v_3}) ~] ~\wedge \nonumber\\ 
& ~[~ vote(\fin{n_2, v_1}) \vee vote(\fin{n_2, v_2}) \vee vote(\fin{n_2, v_3}) ~] ~\wedge \nonumber\\
& ~[~ vote(\fin{n_3, v_1}) \vee vote(\fin{n_3, v_2}) \vee vote(\fin{n_3, v_3}) ~]
\label{eq:ex_clause1_orbit}
\end{align}}%
\section{\symic: Symmetric Incremental Induction}
\label{sec:symclause}
\symic~is an extension of the standard IC3 algorithm~\cite{bradley2011sat,een2011efficient} that takes advantage of the symmetries in a finite instance $\fpro$ of an unbounded protocol $\pro$ to \textit{boost learning} during backward reachability. Specifically, it refines the current frame, in a \textit{single} step, with \emph{all} clauses in the logical orbit $\varphi^{L(G)}$ of a newly-learned quantifier-free clause $\varphi$. In other words, having determined that the backward 1-step check $F_{i-1} \wedge \hat{T} \wedge [\neg \varphi]'$ is unsatisfiable (i.e., that states in cube $\neg \varphi$ in frame $F_i$ are unreachable from the previous frame $F_{i-1}$), \symic~refines $F_i$ with $\varphi^{L(G)}$, i.e., $F_i := F_i \wedge \varphi^{L(G)}$, rather than with just $\varphi$. Thus, at each refinement step, \symic~not only blocks cube $\neg \varphi$, but also all symmetrically-equivalent cubes $[\neg \varphi]^{\gamma}$ for all $\gamma \in G$.
This simple change to the standard incremental induction algorithm significantly improves performance since the extra clauses used 
to refine $F_i$ a) are derived \textit{without} making additional backward 1-step queries, and b) provide stronger refinement in each step of backward reachability leading to faster convergence with fewer counterexamples-to-induction (CTIs). The proof of correctness of symmetry boosting can be found in~\ref{sec:proof_symboost}.

\section{Quantifier Inference}
\label{sec:quantinfer}
The key insight underlying our overall approach is that the explicit logical orbit, in a finite protocol instance, of a learned clause $\varphi$  can be exactly, and systematically, captured by a corresponding quantified predicate $\Phi$. In retrospect, this should not be surprising since symmetry and quantification can be seen as different ways of expressing invariance under permutation of the sort constants in the clause. 
To motivate the connection between symmetry and quantification, consider the following quantifier-free clause from our running example and a proposed quantified predicate that \textit{implicitly} represents its logical orbit:
{
\eqsize
\begin{align}
\varphi_2 =&~ \neg decision(\fin{v_1}) \vee decision(\fin{v_2}) \nonumber \\
\Phi_2 =& ~ \forall X_1, X_2 \in \fin{value_3}:~ (\text{distinct}~X_1~X_2) \to [~ \neg decision(X_1) \vee decision(X_2) ~] \label{eq:ex_quant_clause2}
\end{align}}%
\noindent 
As shown in Table~\ref{tab:quantinfer}, the logical orbit $\varphi_2^{L(G)}$ consists of 6 logically-distinct clauses corresponding to the 6 permutations of the 3 constants of the $\fin{value_3}$ sort. Evaluating $\Phi_2$ by substituting all $3 \times 3=9$ assignments to the variable pair $(X_1, X_2) \in \fin{value_3 \times value_3}$ yields 9 clauses, 3 of which (shown \textcolor{nonecolor}{faded}) are trivially true since their ``distinct'' antecedents are false, with the remaining $6$ corresponding to each of the clauses obtained through permutations of the 3 $\fin{value_3}$ constants.
\begin{table}[!bt]
\begin{center}
\setlength\tabcolsep{8pt}
\tlasize
\resizebox{0.9\textwidth}{!}{
\begin{tabular}{c|c|c}
    \hline
    $(X_1, X_2)$ & Instantiation of $\Phi_2$ & Permutation \\
    \hline
    $(\fin{v_1,v_1})$ & $\text{\textcolor{nonecolor}{$(\text{distinct}~\fin{v_1~v_1}) \to [~ \mathcolorbox{white}{\neg decision(\fin{v_1}) \vee decision(\fin{v_1})}$}} ~]$ &  \textcolor{nonecolor}{none} \\
    $(\fin{v_1,v_2})$ & $(\text{distinct}~\fin{v_1~v_2}) \to [~ \text{$\mathcolorbox{hlcolor}{\neg decision(\fin{v_1}) \vee decision(\fin{v_2})}$} ~]$ &  $()$ \\
    $(\fin{v_1,v_3})$ & $(\text{distinct}~\fin{v_1~v_3}) \to [~ \text{$\mathcolorbox{hlcolor}{\neg decision(\fin{v_1}) \vee decision(\fin{v_3})}$} ~]$ &  $(\fin{v_2~v_3})$ \\
    $(\fin{v_2,v_1})$ & $(\text{distinct}~\fin{v_2~v_1}) \to [~ \text{$\mathcolorbox{hlcolor}{\neg decision(\fin{v_2}) \vee decision(\fin{v_1})}$} ~]$ & $(\fin{v_1~v_2})$  \\
    $(\fin{v_2,v_2})$ & $\text{\textcolor{nonecolor}{$(\text{distinct}~\fin{v_2~v_2}) \to [~ \mathcolorbox{white}{\neg decision(\fin{v_2}) \vee decision(\fin{v_2})}$}} ~]$ &  \textcolor{nonecolor}{none} \\
    $(\fin{v_2,v_3})$ & $(\text{distinct}~\fin{v_2~v_3}) \to [~ \text{$\mathcolorbox{hlcolor}{\neg decision(\fin{v_2}) \vee decision(\fin{v_3})}$} ~]$ &  $(\fin{v_1~v_2~v_3})$ \\
    $(\fin{v_3,v_1})$ & $(\text{distinct}~\fin{v_3~v_1}) \to [~ \text{$\mathcolorbox{hlcolor}{\neg decision(\fin{v_3}) \vee decision(\fin{v_1})}$} ~]$ & $(\fin{v_1~v_3~v_2})$  \\
    $(\fin{v_3,v_2})$ & $(\text{distinct}~\fin{v_3~v_2}) \to [~ \text{$\mathcolorbox{hlcolor}{\neg decision(\fin{v_3}) \vee decision(\fin{v_2})}$} ~]$ &  $(\fin{v_1~v_3})$ \\
    $(\fin{v_3,v_3})$ & $\text{\textcolor{nonecolor}{$(\text{distinct}~\fin{v_3~v_3}) \to [~ \mathcolorbox{white}{\neg decision(\fin{v_3}) \vee decision(\fin{v_3})}$}} ~]$ &  \textcolor{nonecolor}{none} \\
	\hline
\end{tabular}
}
\captionsetup{justification=centering,font=small}
\caption{Correlation between symmetry and quantification for $\Phi_2$ from (\ref{eq:ex_quant_clause2}) \\ 
\colorbox{hlcolor}{Highlighted} clauses represent the logical orbit $\varphi_2^{L(G)}$ \\
\textcolor{nonecolor}{none} indicates the clause has no corresponding permutation $\gamma \in Sym(\fin{value_3})$}
\label{tab:quantinfer}
\end{center}
\end{table}
Similarly, we can show that the 3-clause logical orbit $\varphi_1^{L(G)}$ in  (\ref{eq:ex_clause1_orbit}) can be succinctly expressed by the quantified predicate:
{
\eqsize
\begin{equation}
\Phi_1 =~ \forall Y \in \fin{node_3},~\exists X \in \fin{value_3}:~ vote(Y, X) \label{eq:ex_quant_clause1}
\end{equation}}%
which employs universal \textit{and} existential quantification.
And, finally, $\varphi_3$ and $\Phi_3$ below illustrate how a clause whose logical orbit is just itself can also be expressed as an existentially-quantified predicate.
{
\eqsize
\begin{align}
\varphi_3 =&~ decision(\fin{v_1}) \vee decision(\fin{v_2}) \vee decision(\fin{v_3}) \nonumber\\
\Phi_3 =&~ \exists~X \in \fin{value_3}:~ decision(X) \label{eq:ex_quant_clause3_new}
\end{align}}%
We will first describe basic quantifier inference for protocols with independent sorts. This is done by analyzing the syntactic structure of each quantifier-free clause learned during incremental induction to derive a quantified form that expresses the clause's logical orbit. We later discuss extensions to this approach that consider protocols with dependent sorts, such as $ToyConsensus$, for which the basic single-clause quantifier inference may be insufficient. 

\subsection{Basic Quantifier Inference}
\label{sec:quantinfer_main}
Given a quantifier-free clause $\varphi$, quantifier inference seeks to derive a \textit{compact} quantified predicate that \textit{implicitly} represents, rather than explicitly enumerates, its logical orbit.
The procedure must satisfy the following conditions:
\begin{itemize}[leftmargin=15pt]
    \item[] \textit{Correctness} -- The inferred quantified predicate $\Phi$ should be logically-equivalent to the explicit logical orbit $\varphi^{L(G)}$.
    \item[] \textit{Compactness} -- The number of quantified variables in $\Phi$ for each sort $\s \in S$ should be independent of the sort size $|\s|$. Intuitively, this condition ensures that the size of the quantified predicate, measured as the number of its quantifiers, remains bounded for \textit{any} finite protocol instance, and more importantly, for the unbounded protocol.
\end{itemize}
\noindent \symic~constructs the orbit's quantified representation by a) inferring the required quantifiers for each sort separately, and b) stitching together the inferred quantifiers for the different sorts to form the final result.
The key to capturing the logical orbit and deriving its compact quantified representation is a simple analysis of the \emph{structural distribution} of each sort's constants in the target clause.
Let $\pi(\varphi, \s)$ be a partition of the constants of sort $\s$ in $\varphi$ based on whether or not they appear \textit{identically} in the literals of $\varphi$. Two constants $\fin{c_i}$ and $\fin{c_j}$ are identically-present in $\varphi$ if they occur in $\varphi$ and swapping them results in a logically-equivalent clause, i.e., $\varphi^{(\fin{c_i~c_j})} \same \varphi$.
Let $\numc(\varphi, \s)$ be the number of constants of $\s$ that appear in $\varphi$, and let $|\pi(\varphi, \s)|$ be the number of classes/cells in $\pi(\varphi, \s)$.
Consider the following scenarios for quantifier inference on sort $\s$: 

\vspace{3pt}
\begin{enumerate}[left= -2pt, label=A. ]
    \item \textbf{$\numc(\varphi, \s) < |\s|$ (infer $\forall$)}
\end{enumerate}
\noindent
In this case, clause $\varphi$ contains a strict subset of constants from sort $\s$, indicating that the number of literals in $\varphi$ parameterized by $\s$ constants is \textit{independent} of the sort size $|\s|$. Increasing sort size simply makes the orbit \textit{longer} by adding more symmetrically-equivalent but logically-distinct clauses.
An example of this case is $\varphi_2$ and $\Phi_2$ in (\ref{eq:ex_quant_clause2}). The quantified predicate representing such an orbit requires \textbf{$\numc(\varphi, \s)$} universally-quantified sort variables corresponding to the \textbf{$\numc(\varphi, \s)$} sort constants in the clause, and expresses the orbit as an implication whose antecedent is a ``distinct'' constraint that ensures that the variables cannot be instantiated with identical constants.

\vspace{3pt}
\begin{enumerate}[left= -2pt, label=B. ]
    \item \textbf{$\numc(\varphi, \s) = |\s|$}
\end{enumerate}
\noindent
When all constants of a sort $\s$ appear in a clause, the above universal quantification yields a predicate with $|\s|$ quantified variables and fails the compactness requirement since the number of quantified variables becomes unbounded as the sort size increases. Correct quantification in this case must be inferred by examining the partition of the sort constants in the clause.

\vspace{3pt}
\begin{enumerate}[leftmargin=0pt, label= ]
    \item I. \textbf{Single-cell Partition i.e., $|\pi(\varphi, \s)| = 1$ (infer $\exists$)} 
\end{enumerate}
When all sort constants appear \textit{identically} in $\varphi$,  $\pi(\varphi, \s)$ is a unit partition. Applying \textit{any} permutation $\gamma \in Sym(\s)$ to $\varphi$ yields a logically-equivalent clause, i.e., the logical orbit in this case is just a single clause. Increasing the size of sort $\s$ simply yields a \textit{wider} clause and suggests that such an orbit can be encoded as a predicate with a single existentially-quantified variable that ranges over all the sort constants.
For example, the partition of the $\fin{value_3}$ sort constants in $\varphi_1$ from (\ref{eq:ex_clause1}) is $\pi(\varphi_1, \fin{value_3}) = \{ \fin{\{v_1,v_2,v_3}\}\}$ since all three constants appear identically in $\varphi_1$. 
The orbit of this clause is just itself and can be encoded as:
{
\eqsize
$$\Phi_1(\fin{value_3}) =~ \exists X \in \fin{value_3}:~ vote(\fin{n_1}, X)$$
}%
Also, since $\numc(\varphi_1, \fin{node_3}) < \fin{|node_3|}$, universal quantification (as in Section~\ref{sec:quantinfer_main}.A) correctly captures the dependence of the clause's logical orbit on the $\fin{node_3}$ sort  to get the overall quantified predicate $\Phi_1$ in (\ref{eq:ex_quant_clause1}).

\vspace{3pt}
\begin{enumerate}[leftmargin=0pt, label= ]
    \item II. \textbf{Multi-cell Partition  i.e., $|\pi(\varphi, \s)| > 1$ (infer $\forall\exists$)}
\end{enumerate}
In this case, a fixed number of the constants of sort $\s$ appear differently in $\varphi$ with the remaining constants appearing identically, resulting  in a multi-cell partition. 
Specifically, assume that a number $0 < k < |\s|$ exists that is independent of  $|\s|$ such that  $\pi(\varphi, \s)$ has $k+1$ cells in which one cell has $|\s| - k$ identically-appearing constants and each of the remaining $k$ cells contains one  of the differently-appearing constants.  It can be shown that the logical orbit in this case can be expressed by a quantified predicate with $k$
universal quantifiers and a single existential quantifier.
For example, the partition of the $\fin{value_3}$ constants in the clause:
{
\eqsize
$$\varphi_4 = ~ \neg decision(\fin{v_1}) \vee decision(\fin{v_2}) \vee decision(\fin{v_3})$$
}%
is $\pi(\varphi_4, \fin{value_3}) = \{ \fin{ \{v_1\}, \{v_2, v_3\}} \}$ since $\fin{v_1}$ appears differently from $\fin{v_2}$ and $\fin{v_3}$. 
The logical orbit of this clause is:
{
\eqsizesmall
\begin{align}
\varphi_4^{L(G)} =
& ~[~ \neg decision(\fin{v_1}) \vee decision(\fin{v_2}) \vee decision(\fin{v_3}) ~] ~\wedge \nonumber\\ 
& ~[~ \neg decision(\fin{v_2}) \vee decision(\fin{v_1}) \vee decision(\fin{v_3}) ~] ~\wedge \nonumber\\ 
& ~[~ \neg decision(\fin{v_3}) \vee decision(\fin{v_2}) \vee decision(\fin{v_1}) ~] \label{eq:ex_clause3}
\end{align}}%
and can be compactly encoded with an outer universally-quantified variable corresponding to the sort constant in the singleton cell, and an inner existentially-quantified variable corresponding to the other $|\fin{value_3}| - 1$ identically-present sort constants. A ``distinct'' constraint must also be conjoined with the literals involving the existentially-quantified variable to exclude the constant corresponding to the universally-quantified variable from the inner quantification. 
$\varphi_4^{L(G)}$ can thus be shown to be logically-equivalent to:
{
\eqsize
\begin{equation}
\resizebox{0.925\textwidth}{!}{
$\Phi_4 = \forall Y \in \fin{value_3},~\exists X \in \fin{value_3}:~ \neg decision(Y) \vee \left[ (\text{distinct}~Y ~X) \wedge decision(X) \right]$
} \label{eq:ex_quant_clause3}
\end{equation}
}%

\paragraph{Combining Quantifier Inference for Different Sorts---\hspace{10pt}}
The complete quantified predicate $\Phi$ representing the logical orbit of clause $\varphi$ can be obtained by applying the above inference procedure to each sort in $\varphi$ separately and in any order. This is possible since the sorts are assumed to be independent: the constants of one sort do not permute with the constants of a different sort.
This will yield a predicate $\Phi$ that has the quantified prenex form $\forall^*\exists^* <\fin{CNF~expression}>$, where all universals for each sort are collected together and precede all the existential quantifiers.

It is interesting to note that this connection between symmetry and quantification suggests that an orbit can be visualized as a two-dimensional object whose height and width correspond, respectively, to the number of universally- and existentially-quantified variables. A proof of the correctness of this quantifier inference procedure can be found in~\ref{sec:proof_qi}.

\subsection{Quantifier Inference Beyond $\forall^*\exists^*$}
\label{sec:quantinfer_extensions}
We observed that for some protocols, particularly those that have dependent sorts such as \textit{ToyConsensus}, the above inference procedure violates the compactness requirement. In other words, restricting inference to a $\forall^*\exists^*$ quantifier prefix causes the number of quantifiers to become unbounded as sort sizes increase. Recalling that the $\forall^*\exists^*$ pattern is inferred from the symmetries of a \textit{single} clause, whose literals are the protocol's state variables, suggests that inference of more complex quantification patterns may necessitate that we examine the structural distribution of sort constants across \textit{sets of clauses.} While this is an interesting possible direction for further exploration of the connection between symmetry and quantification, an alternative approach is to take advantage of the \textit{formula structure} of the protocol's transition relation. For example, the transition relation of \textit{ToyConsensus} is specified in terms of two quantified sub-formulas, $didNoteVote$ and $chosenAt$, that can be viewed, in analogy with a sequential hardware circuit, as internal auxiliary non-state variables that act as ``combinational'' functions of the state variables. By allowing such auxiliary variables to appear explicitly in clauses learned during incremental induction, the quantified predicates representing the logical orbits of these clauses (according to the basic inference procedure in Section~\ref{sec:quantinfer_main}) will \textit{implicitly} incorporate the quantifiers used in the auxiliary variable definitions and automatically have a quantifier prefix that generalizes the basic $\forall^*\exists^*$ template. 

\paragraph{Revisiting ToyConsensus---\hspace{10pt}}
\noindent
When \symic is run on the finite instance \naiveinstance, it terminates with the following two strengthening assertions:
{
\eqsizesmall
\begin{align}
& \scalebox{0.9}{\text{$A_1 =~ \forall~{N \in \fin{node_3}, V_1, V_2 \in \fin{value_3}}:~(\text{distinct}~V_1~V_2) \to \neg vote(N, V_1) \vee \neg vote(N, V_2)$}} \label{eq:a1}\\
& \scalebox{0.9}{\text{$A_2 =~ \forall~{V \in \fin{value_3}},~\exists~{Q \in \fin{quorum_3}}.~\neg decision(V) \vee chosenAt(Q, V)$}}  \label{eq:a2}\\
& \scalebox{0.9}{\text{$\hphantom{A_2} =~ \forall~{V \in \fin{value_3}},~\exists~{Q \in \fin{quorum_3}}.~\neg decision(V) \vee [~\forall~{N \in Q}: vote(N, V)~]$}} \nonumber
\end{align}
}%
which, together with $\hat{P}$, serve as an inductive invariant proving that $\hat{P}$ holds for this instance. Both assertions are obtained using the basic quantifier inference procedure in Section~\ref{sec:quantinfer_main} that produces a $\forall^*\exists^*$ quantifier prefix in terms of the clause variables. Note, however, that $A_2$ is expressed in terms of the auxiliary variable $chosenAt$. Substituting the definition of $chosenAt$ yields an assertion with a $\forall\exists\forall$ quantifier prefix exclusively in terms of the protocol's state variables.

\section{Finite Convergence Checks}
\label{sec:invcheck}
\noindent Given a safe finite instance $\fpro \triangleq \pro(\fin{|s_1|,\dots,|s_n|})$, let  $Inv_{\lo{|s_1|,\dots,|s_n|}}$ denote the inductive invariant derived by \symic to prove that $\hat{P}$ holds in $\fpro$. What remains is to determine the instance size $\fin{|s_1|, \dots, |s_n|}$ needed so that $Inv_{\lo{|s_1|,\dots,|s_n|}}$ is also an inductive invariant for all sizes.
If the instance size is too small, $\fpro$ may not include all protocol behaviors and $Inv_{\lo{|s_1|,\dots,|s_n|}}$ will not be inductive at larger sizes.
As shown in the invisible invariant approach~\cite{pnueli2001automatic, arons2001parameterized, zuck2004model,balaban2005iiv,namjoshi2007symmetry}, increasing the instance size becomes necessary to include new protocol behaviors missing in $\fpro$, until protocol behaviors \textit{saturate}.
We propose an \textit{automatic} way to update the instance size and reach saturation by starting with an initial $\szinit$ and iteratively increasing the size until \textit{finite convergence} is achieved.

The initial base size can be chosen to be any non-trivial instance size and can be easily determined by a simple analysis of the protocol description. For example, any non-trivial instance of the \naiveconsensus protocol should have $|\fin{node}| \geq 3$, $|\fin{quorum}| \geq 3$, and $|\fin{value}| \geq 2$.

Our finite convergence procedure can be seen as an integration of symmetry saturation and a stripped-down form of multi-dimensional mathematical induction, and has similarities with previous works on structural induction~\cite{kurshan1989structural,german1992reasoning} and proof convergence~\cite{dooley2016proving}. To determine if $Inv_{\lo{|s_1|,\dots,|s_n|}}$ is inductive for any size, the procedure performs the following checks for $1 \leq i \leq \fin{n}$:
{
\eqsize
\begin{align}
& \scalebox{0.89}{\text{a) $Init(\fin{|s_1|}..\text{\hlc{$\fin{|s_i|+1}$}}..\fin{|s_n|}) \to Inv_{\lo{|s_1|,\dots,|s_n|}}(\fin{|s_1|}..\text{\hlc{$\fin{|s_i|+1}$}}..\fin{|s_n|})$}} \label{eq:induct_check1}\\
& \scalebox{0.89}{\text{b) $Inv_{\lo{|s_1|,\dots,|s_n|}}(\fin{|s_1|}..\text{\hlc{$\fin{|s_i|+1}$}}..\fin{|s_n|}) \wedge T(\fin{|s_1|}..\text{\hlc{$\fin{|s_i|+1}$}}..\fin{|s_n|}) \to Inv'_{\lo{|s_1|,\dots,|s_n|}}(\fin{|s_1|}..\text{\hlc{$\fin{|s_i|+1}$}}..\fin{|s_n|})$}} \label{eq:induct_check2}
\end{align}}%
where $Inv_{\lo{|s_1|,\dots,|s_n|}}(\fin{|s_1|}..\text{\hlc{$\fin{|s_i|+1}$}}..\fin{|s_n|})$ denotes the application of $Inv_{\lo{|s_1|,\dots,|s_n|}}$ to an instance in which the size of sort $\fin{s_i}$ is increased by 1 while the sizes of the other sorts are unchanged.\footnote{Sort dependencies, if any, should be considered when increasing a sort size.}

If all of these checks pass, we can conclude that $Inv_{\lo{|s_1|,\dots,|s_n|}}$ is not specific to the instance size used to derive it and that we have reached \textit{cutoff}, i.e., that $Inv_{\lo{|s_1|,\dots,|s_n|}}$ is an inductive invariant for \textit{any} size.
Intuitively, this suggests that adding a new protocol component (e.g., client, server, node, proposer, acceptor) does not add any unseen unique behavior, and hence proving safety till the cutoff is sufficient to prove safety for any instance size.
While we believe these checks are sufficient, we still do not have a formal convergence proof. In our implementation, we confirm convergence by performing the unbounded induction checks a) $Init \to Inv_{\lo{|s_1|,\dots,|s_n|}}$, and b) $Inv_{\lo{|s_1|,\dots,|s_n|}} \wedge T \to Inv_{\lo{|s_1|,\dots,|s_n|}}'$ noting that they may lie outside the decidable fragment of first-order logic.

On the other hand, failure of these checks, say for sort $\fin{s_i}$, implies that $Inv_{\lo{|s_1|\dots|s_n|}}$ will fail for larger sizes and cannot be inductive in the unbounded case, and we need to repeat \symic on a finite instance with an increased size for sort $\fin{s_i}$, i.e., $\fpro_{new} \triangleq \pro(\fin{|s_1|},..,\text{\hlc{$\fin{|s_i|+1}$}},..,\fin{|s_n|})$, to include new protocol behaviors that are missing in $\fpro$.

Recall from (\ref{eq:a1}) and (\ref{eq:a2}), running \symic on $ToyConsensus(3,3,3)$ produces $Inv_{\lo{3,3,3}} = A_1 \wedge A_2 \wedge \hat{P}$. $Inv_{\lo{3,3,3}}$ passes checks (\ref{eq:induct_check1}) and (\ref{eq:induct_check2}) for instances $ToyConsensus(4,4,3)$ and $ToyConsensus(3,3,4)$, indicating finite convergence.\footnote{Since $\fin{quorum}$ is a dependent sort on $\fin{node}$, it is increased together with the $\fin{node}$ sort.} $Inv_{\lo{3,3,3}}$ passes standard induction checks in the unbounded domain as well, establishing it as a proof certificate that proves the property as safe in $ToyConsensus$.

\section{\icpo: IC3 for Proving Protocol Properties}
\label{sec:ic3po_highlevel}
\noindent
Given a protocol specification $\pro$, \icpo iteratively invokes \symic~on finite instances of increasing size, starting with a given initial base size. Upon termination, \icpo~either a) reaches convergence on an inductive invariant $Inv_{\lo{|s_1|,\dots,|s_n|}}$ that proves  $P$ for the unbounded protocol $\pro$, or b) produces a counterexample trace $Cex_{\lo{|s_1|,\dots,|s_n|}}$ that serves as a finite witness to its violation in both the finite instance and the unbounded protocol. The detailed pseudo code of \icpo is available in~\ref{app:algo}.

We also explored a number of simple enhancements to \icpo, such as strengthening the inferred quantified predicates whenever safely possible to do during incremental induction by a) dropping the ``distinct'' antecedent, and b) rearranging the quantifiers if the strengthened predicate is still unreachable from the previous frame. We describe these enhancements in~\ref{sec:opt}. The results presented in this paper were obtained without these enhancements.

\paragraph{Implementation---\hspace{10pt}}
\noindent
Our implementation of \icpo is publicly available at~\url{https://github.com/aman-goel/ic3po}. The implementation accepts protocol descriptions in the Ivy language~\cite{padon2016ivy} and uses the Ivy compiler to extract a quantified, logical formulation $\pro$ in a customized VMT~\cite{vmt} format.
We use a modified version~\cite{gitpysmt} of the pySMT~\cite{gario2015pysmt} library to implement our prototype, and use the Z3~\cite{demoura2008z3} solver for all SMT queries. We use the SMT-LIB~\cite{BarFT-SMTLIB} theory of free sorts and function symbols with datatypes and quantifiers ({\tt UFDT}), which allows formulating SMT queries for both, the finite and the unbounded domains. 
For a safe protocol, the inductive proof is printed in the Ivy format as an \textit{independently check-able} proof certificate, which can be further validated with the Ivy verifier.

\vspace{-5pt}
\section{Evaluation}
\label{sec:evaluation}
\noindent
We evaluated \icpo~on a total of $29$ distributed protocols including 4 problems from~\cite{ma2019i4}, 13 from~\cite{pldi20folic3}, and 12 from~\cite{ivybench}. 
This evaluation set includes fairly complex models of consensus algorithms as well as protocols such as two-phase commit, chord ring, hybrid reliable broadcast, etc. Several studies~\cite{hawblitzel2015ironfleet,padon2016ivy,ma2019i4,pldi20folic3,feldman2019inferring,berkovits2019verification} have indicated the challenges involved in verifying these protocols. 

All $29$ protocols are safe based on manual verification. Even though finding counterexample traces is equally important, we limit our evaluation to safe protocols where the property holds, since inferring inductive invariants is the main bottleneck of existing techniques for verifying distributed protocols~\cite{feldman2019complexity,padon2016ivy, DBLP:journals/corr/abs-2008-09909}.

We compared \icpo against the following $3$ verifiers that implement state-of-the-art IC3-style techniques for automatic verification of distributed protocols:
\begin{itemize}
    \item \ifour~\cite{ma2019i4} performs finite-domain IC3 (without accounting for symmetry) using the AVR model checker~\cite{goel2020avr}, followed by iteratively generalizing and checking the inductive invariant produced by AVR using Ivy.
    \item \updr is the implementation of the PDR$^{\forall}$/UPDR algorithm~\cite{10.1145/3022187} for verifying distributed protocols, from the \textit{mypyvy}~\cite{mypyvygithub} framework.
    \item \folic~\cite{pldi20folic3} is a recent technique implemented in \textit{mypyvy} that extends IC3 with the ability to infer inductive invariants with quantifier alternations.
\end{itemize}
All experiments were performed on an Intel~(R) Xeon CPU (X5670). For each run, we used a timeout of 1 hour and a memory limit of 32 GB. All tools were executed in their respective default configurations. We used Z3~\cite{demoura2008z3} version 4.8.9, Yices 2~\cite{dutertre2014yices} version 2.6.2, and CVC4~\cite{BCD_11} version 1.7.

\subsection{Results}
\label{sec:eval_results}
\begin{table}[!tb]
\setlength\tabcolsep{2pt}
\begin{center}
\resizebox{\textwidth}{!}{
\begin{tabular}{l|r|lrrr|rrr|rrr|rrr}
    \multicolumn{1}{c}{} & \multicolumn{1}{c}{\makebox[20pt][c]{\textit{Human}}} & \multicolumn{4}{c}{\icpo} & \multicolumn{3}{c}{\ifour} & \multicolumn{3}{c}{\updr} & \multicolumn{3}{c}{\folic} \\
    \hline
    \multicolumn{1}{c}{Protocol (\#29)} & \multicolumn{1}{|c}{Inv} & \multicolumn{1}{|c}{info} & \multicolumn{1}{c}{Time} & \multicolumn{1}{c}{Inv} & \multicolumn{1}{c}{SMT} & \multicolumn{1}{|c}{Time} & \multicolumn{1}{c}{Inv} & \multicolumn{1}{c}{SMT} & \multicolumn{1}{|c}{Time} & \multicolumn{1}{c}{Inv} & \multicolumn{1}{c}{SMT} & \multicolumn{1}{|c}{Time} & \multicolumn{1}{c}{Inv} & \multicolumn{1}{c}{SMT} \\
    
	\hline
	tla-consensus  & 1 & \eprsym & \textbf{0} & 1 & 17 & 4 & 1 & 7 & 0 & 1 & 38 & 1 & 1 & 29 \\
	tla-tcommit & 3 & \eprsym & \textbf{1} & 2 & 31 &  \multicolumn{2}{l}{ \textcolor{failcolor}{unknown} } & 71 & 1 & 3 & 214 & 2 & 3 & 162 \\
	i4-lock-server & 2 & \eprsym & \textbf{1} & 2 & 37 & 2 & 2 & 35 & 1 & 2 & 133 & 1 & 2 & 66 \\
	ex-quorum-leader-election & 3 & \eprsym & \textbf{3} & 5 & 129 & 32 & 14 & 15429 & 11 & 3 & 1007 & 24 & 8 & 1078 \\
	pyv-toy-consensus-forall & 4 & \eprsym & \textbf{3} & 4 & 105 &  \multicolumn{2}{l}{ \textcolor{failcolor}{unknown} } & 5949 & 10 & 3 & 590 & 11 & 5 & 587 \\
	tla-simple & 8 & \eprsym & 6 & 3 & 285 & \textbf{4} & 3 & 1319 &  \multicolumn{2}{l}{ \textcolor{failcolor}{timeout} } & &  \multicolumn{2}{l}{ \textcolor{failcolor}{timeout} } & \\
	ex-lockserv-automaton & 2 & \eprsym & 7 & 12 & 594 & \textbf{3} & 15 & 1731 & 21 & 9 & 3855 & 10 & 12 & 1181 \\
	tla-simpleregular & 9 & \eprsym & \textbf{8} & 4 & 346 &  \multicolumn{2}{l}{ \textcolor{failcolor}{unknown} } & 14787 &  \multicolumn{2}{l}{ \textcolor{failcolor}{timeout} } & & 57 & 9 & 314 \\
	pyv-sharded-kv & 5 & \eprsym & 10 & 8 & 590 & \textbf{4} & 15 & 2101 & 6 & 7 & 784 & 22 & 10 & 522 \\
	pyv-lockserv & 9 & \eprsym & 11 & 12 & 702 & \textbf{3} & 15 & 1606 & 14 & 9 & 3108 & 8 & 11 & 1044 \\
	tla-twophase & 12 & \eprsym & 14 & 10 & 984 &  \multicolumn{2}{l}{ \textcolor{failcolor}{unknown} } & 10505 & 67 & 14 & 12031 & \textbf{9} & 12 & 1635 \\
	i4-learning-switch & 8 & \eprsym & \textbf{14} & 9 & 589 & 22 & 11 & 26345 &  \multicolumn{2}{l}{ \textcolor{failcolor}{timeout} } & &  \multicolumn{2}{l}{ \textcolor{failcolor}{timeout} } & \\
	ex-simple-decentralized-lock & 5 & \eprsym & 19 & 15 & 2219 & 14 & 22 & 5561 & 4 & 2 & 677 & \textbf{4} & 8 & 291 \\
	i4-two-phase-commit & 11 & \eprsym & 27 & 11 & 2541 & \textbf{4} & 16 & 4045 & 16 & 9 & 2799 & 8 & 9 & 1083 \\
	pyv-consensus-wo-decide & 5 & \eprsym & \textbf{50} & 9 & 1886 & 1144 & 42 & 41137 & 100 & 4 & 8563 & 168 & 26 & 5692 \\
	pyv-consensus-forall & 7 & \eprsym & \textbf{99} & 10 & 3445 & 1006 & 44 & 156838 & 490 & 6 & 24947 & 2461 & 27 & 16182 \\
	pyv-learning-switch & 8 & \eprsym & \textbf{127} & 13 & 3388 & 387 & 49 & 51021 & 278 & 11 & 3210 &  \multicolumn{2}{l}{ \textcolor{failcolor}{timeout} } & \\
	i4-chord-ring-maintenance & 18 & \eprsym & \textbf{229} & 12 & 6418 &  \multicolumn{2}{l}{ \textcolor{failcolor}{timeout} } & &  \multicolumn{2}{l}{ \textcolor{failcolor}{timeout} } & &  \multicolumn{2}{l}{ \textcolor{failcolor}{timeout} } & \\
	pyv-sharded-kv-no-lost-keys & 2 & \altsym \eprsym & \textbf{3} & 2 & 57 &  \multicolumn{2}{l}{ \textcolor{failcolor}{unknown} } & 1232 &  \multicolumn{2}{l}{ \textcolor{failcolor}{unknown} } & 73 & 3 & 2 & 51 \\
	ex-naive-consensus & 4 & \altsym \eprsym & \textbf{6} & 4 & 239 &  \multicolumn{2}{l}{ \textcolor{failcolor}{unknown} } & 15141 &  \multicolumn{2}{l}{ \textcolor{failcolor}{unknown} } & 1325 & 73 & 18 & 414 \\
	pyv-client-server-ae & 2 & \altsym \defsym \eprsym & \textbf{2} & 2 & 49 &  \multicolumn{2}{l}{ \textcolor{failcolor}{unknown} } & 1483 &  \multicolumn{2}{l}{ \textcolor{failcolor}{unknown} } & 132 & 877 & 15 & 700 \\
	ex-simple-election & 3 & \altsym \defsym \eprsym & \textbf{7} & 4 & 268 &  \multicolumn{2}{l}{ \textcolor{failcolor}{unknown} } & 2747 &  \multicolumn{2}{l}{ \textcolor{failcolor}{unknown} } & 1147 & 32 & 10 & 222 \\
	pyv-toy-consensus-epr & 4 & \altsym \defsym \eprsym & \textbf{9} & 4 & 370 &  \multicolumn{2}{l}{ \textcolor{failcolor}{unknown} } & 5944 &  \multicolumn{2}{l}{ \textcolor{failcolor}{unknown} } & 473 & 70 & 14 & 217 \\
	ex-toy-consensus & 3 & \altsym \defsym \eprsym & \textbf{10} & 3 & 209 &  \multicolumn{2}{l}{ \textcolor{failcolor}{unknown} } & 2797 &  \multicolumn{2}{l}{ \textcolor{failcolor}{unknown} } & 348 & 21 & 8 & 124 \\
	pyv-client-server-db-ae & 5 & \altsym \defsym \eprsym & \textbf{17} & 6 & 868 &  \multicolumn{2}{l}{ \textcolor{failcolor}{unknown} } & 81509 &  \multicolumn{2}{l}{ \textcolor{failcolor}{unknown} } & 422 &  \multicolumn{2}{l}{ \textcolor{failcolor}{timeout} } & \\
	pyv-hybrid-reliable-broadcast & 8 & \altsym \defsym \eprsym & \textbf{587} & 4 & 1474 &  \multicolumn{2}{l}{ \textcolor{failcolor}{unknown} } & 34764 &  \multicolumn{2}{l}{ \textcolor{failcolor}{unknown} } & 713 & 1360 & 23 & 3387 \\
	pyv-firewall & 2 & \altsym \ndefsym \neprsym & \textbf{2} & 3 & 131 &  \multicolumn{2}{l}{ \textcolor{failcolor}{unknown} } & 344 &  \multicolumn{2}{l}{ \textcolor{failcolor}{unknown} } & 130 & 7 & 8 & 116 \\
	ex-majorityset-leader-election & 5 & \altsym \ndefsym \neprsym & \textbf{72} & 7 & 1552 &  \multicolumn{2}{l}{ \textcolor{failcolor}{error} } & &  \multicolumn{2}{l}{ \textcolor{failcolor}{unknown} } & 2350 &  \multicolumn{2}{l}{ \textcolor{failcolor}{timeout} } & \\
	pyv-consensus-epr & 7 & \altsym \defsym \neprsym & \textbf{1300} & 9 & 29601 &  \multicolumn{2}{l}{ \textcolor{failcolor}{unknown} } & 177189 &  \multicolumn{2}{l}{ \textcolor{failcolor}{unknown} } & 7559 & 1468 & 30 & 3355 \\
	\hline
	\multicolumn{2}{l}{No. of problems solved \hfill (out of 29)} & \multicolumn{4}{|c}{\textbf{29}} & \multicolumn{3}{|c}{13} & \multicolumn{3}{|c}{14} & \multicolumn{3}{|c}{23}\\
	\multicolumn{2}{l}{Uniquely solved} & \multicolumn{4}{|c}{\textbf{3}} & \multicolumn{3}{|c}{0} & \multicolumn{3}{|c}{0} & \multicolumn{3}{|c}{0}\\
	\hline
	\multicolumn{2}{l}{For $10$ cases solved by all: \hfill $\sum$ Time } & \multicolumn{4}{|c}{\textbf{232}} & \multicolumn{3}{|c}{2221} & \multicolumn{3}{|c}{667} & \multicolumn{3}{|c}{2711}\\
	\multicolumn{2}{l}{\hfill $\sum$ Inv \hskip 5pt ~} & \multicolumn{4}{|c}{85} & \multicolumn{3}{|c}{186} & \multicolumn{3}{|c}{\textbf{52}} & \multicolumn{3}{|c}{114}\\
	\multicolumn{2}{l}{\hfill $\sum$ SMT } & \multicolumn{4}{|c}{\textbf{12160}} & \multicolumn{3}{|c}{228490} & \multicolumn{3}{|c}{45911} & \multicolumn{3}{|c}{27168}\\
	\hline
\end{tabular}
}
\captionsetup{justification=centering, belowskip=0pt}
\caption{Comparison of \icpo~against other state-of-the-art verifiers}
\captionsetup{justification=raggedright, aboveskip=0pt,belowskip=0pt}
\caption*{
Time: run time (seconds),
Inv: \# assertions in inductive proof, SMT: \# SMT queries, \\
Column ``info'' provides information on the strengthening assertions (i.e., $A$) in \icpo's inductive proof: \altsym indicates $A$ has quantifier alternations, \defsym means $A$ has definitions, and \neprsym means $A$ adds quantifier-alternation cycles
}
\label{tab:table1}
\end{center}
\end{table}

\noindent
Table~\ref{tab:table1} summarizes the experimental results.
Apart from the number of problems solved, we compared the tools on $3$ metrics: run time in seconds, proof size measured by the number of assertions in the inductive invariant for the unbounded protocol, and the total number of SMT queries made. Each tool uses SMT queries differently (e.g., \ifour~uses {\tt QF\_UF} for finite, {\tt UF} for unbounded). Comparing the number of SMT queries still helps in understanding the run time behavior.

\icpo solved all $29$ problems, while $10$ protocols were solved by all the tools. The $5$ rows at the bottom of Table~\ref{tab:table1} provide a summary of the comparison. Overall, compared to the other tools 
\icpo is faster, requires fewer SMT queries, and produces shorter inductive proofs even for problems requiring inductive invariants with quantifier alternations (marked with \altsym in Table~\ref{tab:table1}).

We did a more extensive comparison between the two finite-domain incremental induction verifiers \icpo and \ifour (\ref{sec:eval_symboost}), performed a statistical analysis using multiple runs with different solver seeds to account for the effect of randomness in SMT solving (\ref{sec:eval_statistical}), compared the inductive proofs produced by \icpo~against human-written invariants (\ref{sec:eval_human}), and performed a preliminary exploration of distributed protocols with totally-ordered domains and ring topologies (\ref{sec:eval_order}).

\subsection{Discussion}
\label{sec:eval_discussion}
Comparing \icpo~and \ifour~clearly reveals the benefits of symmetric incremental induction. For example, \ifour requires $7814$ SMT queries to eliminate $443$ CTIs when solving \naiveinstance, compared to $192$ SMT calls and $13$ CTIs  for \icpo.
Even though both techniques perform finite incremental induction, symmetry-aware clause boosting in \icpo leads to a factorial reduction in the number of SMT queries and yields compact inductive proofs.

Comparing \icpo and \updr reveals the benefits of finite-domain reasoning methods compared to direct unbounded verification. Even in cases where existential quantifier inference isn't necessary,  symmetry-aware finite-domain reasoning gives \icpo an edge both in terms of run time and the number of SMT queries. 

Comparing \icpo and \folic, the only two verifiers that can infer invariants with a combination of universal and existential quantifiers, highlights the advantage of \icpo's approach over the separators-based technique~\cite{pldi20folic3} used in \folic. The significant performance edge that \icpo has over \folic is due to the fact that a) reasoning in \icpo is primarily in a (small) finite domain compared to \folic's unbounded reasoning, and b) unlike \folic which enumeratively searches for specific quantifier patterns, \icpo finds the required invariants without search  by automatically inferring their patterns from the symmetry of the protocol. 

Overall, the evaluation confirms the main hypothesis of this paper, that it is possible to use the relationship between symmetry and quantification to scale the verification of distributed protocols beyond the current state-of-the-art.

\vspace{-4pt}
\section{Related Work}
\label{sec:related}
\noindent
Introduced by Lamport, TLA+ is a widely-used language for the specification and verification of distributed protocols~\cite{newcombe2015amazon,beers2008pre}. The accompanying TLC model checker can perform automatic verification on a finite instance of a TLA+ specification, and can also be configured to employ symmetry to improve scalability. However, TLC is primarily intended as a debugging tool for small finite instances and not as a tool for inferring inductive invariants.

Several manual or semi-automatic verification techniques (e.g., using interactive theorem proving or compositional verification) have been proposed for deriving system-level proofs~\cite{owre1992pvs,chaudhuri2010verifying,hawblitzel2015ironfleet,Wilcox2015Verdi,hoenicke2017thread,v2019pretend}. 
These techniques generally require a deep understanding of the protocol being verified and significant manual effort to guide proof development.
The Ivy~\cite{padon2016ivy} system improves on these techniques by graphically displaying CTIs and interactively asking the user to provide strengthening assertions that can eliminate them.

Verification of parameterized systems using SMT solvers is further explored in MCMT~\cite{ghilardi2010backward}, Cubicle~\cite{conchon2012cubicle}, and paraVerifier~\cite{li2015paraverifier}. Abdulla et al.~\cite{abdulla2016parameterized} proposed \textit{view abstraction} to compute the reachable set for finite instances using forward reachability until cutoff is reached. Our technique builds on these works with the capability to automatically infer the required quantified inductive invariant using the latest advancements in model checking, by combining symmetry-aware clause learning and quantifier inference in finite-domain incremental induction. 
The use of derived/ghost variables has been recognized as important in~\cite{lamport1977proving,owicki1976verifying,namjoshi2007symmetry}. \icpo utilizes protocol structure, namely auxiliary definitions in the protocol specification, to automatically infer inductive invariants with complex quantifier alternations.

Several recent approaches (e.g., \updr~\cite{karbyshev2017property}, QUIC3~\cite{gurfinkel2018quantifiers}, Phase-\updr~\cite{feldman2019inferring}, \folic~\cite{pldi20folic3}) extend IC3/PDR to automatically infer quantified inductive invariants.

Unlike \icpo, these techniques rely heavily on unbounded  SMT solving.

Our work is closest in spirit to FORHULL-N~\cite{dooley2016proving} and I4~\cite{ma2019i4,ma2019towards}.
Similar to \icpo, these techniques perform incremental induction over small finite instances of a parameterized system and employ a generalization procedure that transforms finite-domain proofs to quantified inductive invariants that hold for all parameter values.
Dooley and Somenzi proposed FORHULL-N to verify parameterized reactive systems by running bit-level IC3 and generalizing the learnt clauses into candidate universally-quantified proofs through a process of proof saturation and convex hull computation.
These candidate proofs involve modular linear arithmetic constraints as antecedents in a way such that they approximate the protocol behavior beyond the current finite instance, and their correctness is validated by checking them until the cutoff is reached.
I4 uses an ad hoc generalization procedure to obtain universally-quantified proofs from the finite-domain inductive invariants generated by the AVR model checker~\cite{goel2020avr}.

\section{Conclusions and Future Work}
\label{sec:conclusions}
\noindent
\icpo is, to our knowledge, the first verification system that uses the synergistic relationship between symmetry and quantification to automatically infer the quantified inductive invariants required to prove the safety of symmetric  protocols.
Recognizing that symmetry and quantification are alternative ways of capturing invariance, \icpo extends the incremental induction algorithm to learn clause orbits, and encodes these orbits with corresponding logically-equivalent and compact quantified predicates. 
\icpo employs a systematic procedure to check for finite convergence, and outputs quantified inductive invariants, with both universal and existential quantifiers, that hold for all protocol parameters.
Our evaluation demonstrates that \icpo~significantly is a significant improvement over the current state-of-the-art.

Future work includes exploring methods to utilize the regularity in totally-ordered domains during reachability analysis, investigating techniques to counter undecidability in practical distributed systems verification, and exploring enhancements to further improve the scalability to complex distributed protocols and their implementations.
As a long-term goal, we aim towards automatically inferring inductive invariants for complicated distributed protocols, such as Paxos~\cite{lamport2019part,lamport2001paxos}, by building further on this initial work.

\section*{Data Availability Statement and Acknowledgments}
\addcontentsline{toc}{section}{Data Availability Statement and Acknowledgments}
The software and data sets generated and analyzed during the current study, including all experimental data, evaluation scripts, and \icpo source code are available at \url{https://github.com/aman-goel/nfm2021exp}.
We thank the developers of pySMT~\cite{gario2015pysmt}, Z3~\cite{demoura2008z3}, and Ivy~\cite{padon2016ivy} for making their tools openly available. 
We thank the authors of the I4 project~\cite{ma2019i4} for their help in shaping some of the ideas presented in this paper.



\bibliographystyle{splncs04}
\bibliography{header-standard.bib,cite_database.bib,reference-db.bib}

\begin{thebibliography}{10}
\providecommand{\url}[1]{\texttt{#1}}
\providecommand{\urlprefix}{URL }
\providecommand{\doi}[1]{https://doi.org/#1}

\bibitem{ivyclientserver}
Client server protocol in ivy.
  \url{http://microsoft.github.io/ivy/examples/client_server_example.html}

\bibitem{ivybench}
A collection of distributed protocol verification problems.
  \url{https://github.com/aman-goel/ivybench}

\bibitem{ivy}
The ivy language and verifier. \url{http://microsoft.github.io/ivy}

\bibitem{mypyvygithub}
mypyvy (github). \url{https://github.com/wilcoxjay/mypyvy}

\bibitem{gitpysmt}
{pySMT: A library for SMT formulae manipulation and solving}.
  \url{https://github.com/aman-goel/pysmt}

\bibitem{toyconsensus}
{Toy consensus protocol}.
  \url{https://github.com/microsoft/ivy/blob/master/examples/ivy/toy_consensus.ivy}

\bibitem{abdulla2016parameterized}
Abdulla, P., Haziza, F., Hol{\'\i}k, L.: Parameterized verification through
  view abstraction. International Journal on Software Tools for Technology
  Transfer  \textbf{18}(5),  495--516 (2016)

\bibitem{apt1986limits}
Apt, K.R., Kozen, D.: Limits for automatic verification of finite-state
  concurrent systems. Inf. Process. Lett.  \textbf{22}(6),  307--309 (1986)

\bibitem{arons2001parameterized}
Arons, T., Pnueli, A., Ruah, S., Xu, Y., Zuck, L.: Parameterized verification
  with automatically computed inductive assertions? In: Berry, G., Comon, H.,
  Finkel, A. (eds.) Computer Aided Verification. pp. 221--234. Springer Berlin
  Heidelberg, Berlin, Heidelberg (2001)

\bibitem{balaban2005iiv}
Balaban, I., Fang, Y., Pnueli, A., Zuck, L.D.: Iiv: An invisible invariant
  verifier. In: International Conference on Computer Aided Verification. pp.
  408--412. Springer (2005)

\bibitem{balyo2020proceedings}
Balyo, T., Froleyks, N., Heule, M.J., Iser, M., J{\"a}rvisalo, M., Suda, M.:
  Proceedings of sat competition 2020: Solver and benchmark descriptions
  (2020)

\bibitem{barner2002combining}
Barner, S., Grumberg, O.: Combining symmetry reduction and under-approximation
  for symbolic model checking. In: International Conference on Computer Aided
  Verification. pp. 93--106. Springer (2002)

\bibitem{BCD_11}
Barrett, C., Conway, C.L., Deters, M., Hadarean, L., Jovanovi{'{c}}, D., King,
  T., Reynolds, A., Tinelli, C.: {CVC4}. In: Gopalakrishnan, G., Qadeer, S.
  (eds.) Proceedings of the 23rd International Conference on Computer Aided
  Verification (CAV '11). Lecture Notes in Computer Science, vol.~6806, pp.
  171--177. Springer (Jul 2011),
  \url{http://www.cs.stanford.edu/~barrett/pubs/BCD+11.pdf}, snowbird, Utah

\bibitem{BarFT-SMTLIB}
Barrett, C., Fontaine, P., Tinelli, C.: {The Satisfiability Modulo Theories
  Library (SMT-LIB)}. \url{www.SMT-LIB.org} (2016)

\bibitem{beers2008pre}
Beers, R.: {Pre-RTL formal verification: an intel experience}. In: Proceedings
  of the 45th annual Design Automation Conference. pp. 806--811 (2008)

\bibitem{berkovits2019verification}
Berkovits, I., Lazic, M., Losa, G., Padon, O., Shoham, S.: Verification of
  threshold-based distributed algorithms by decomposition to decidable logics.
  CoRR  \textbf{abs/1905.07805} (2019), \url{http://arxiv.org/abs/1905.07805}

\bibitem{bloem2015decidability}
Bloem, R., Jacobs, S., Khalimov, A., Konnov, I., Rubin, S., Veith, H., Widder,
  J.: Decidability of parameterized verification. Synthesis Lectures on
  Distributed Computing Theory  \textbf{6}(1),  1--170 (2015).
  \doi{10.2200/S00658ED1V01Y201508DCT013}

\bibitem{bradley2011sat}
Bradley, A.R.: {SAT-Based Model Checking without Unrolling}. In: Proceedings of
  the 12th international conference on Verification, model checking, and
  abstract interpretation. pp. 70--87. VMCAI'11, Springer-Verlag, Berlin,
  Heidelberg (2011), \url{http://dl.acm.org/citation.cfm?id=1946284.1946291}

\bibitem{burch1990symbolic}
Burch, J.R., Clarke, E.M., McMillan, K.L., Dill, D.L., Hwang, L.J.: {Symbolic
  Model Checking: $10^{20}$ States and Beyond}. In: Proceedings. Fifth Annual
  IEEE Symposium on Logic in Computer Science. pp. 428--439 (1990)

\bibitem{burch1992symbolic}
Burch, J.R., Clarke, E.M., McMillan, K.L., Dill, D.L., Hwang, L.J.: {Symbolic
  Model checking: $10^{20}$ States and Beyond}. Information and Computation
  \textbf{98}(2),  142--170 (1992)

\bibitem{chaudhuri2010verifying}
Chaudhuri, K., Doligez, D., Lamport, L., Merz, S.: Verifying safety properties
  with the tla+ proof system. In: International Joint Conference on Automated
  Reasoning. pp. 142--148. Springer (2010)

\bibitem{vmt}
Cimatti, A., Roveri, M., Griggio, A., Irfan, A.: {Verification Modulo
  Theories}. \url{http://www.vmt-lib.org} (2011)

\bibitem{conchon2012cubicle}
Conchon, S., Goel, A., Krsti{\'c}, S., Mebsout, A., Za{\"\i}di, F.: Cubicle: A
  parallel smt-based model checker for parameterized systems. In: International
  Conference on Computer Aided Verification. pp. 718--724. Springer (2012)

\bibitem{demoura2008z3}
De~Moura, L., Bj{\o}rner, N.: {Z3: An efficient SMT solver}. In: Tools and
  Algorithms for the Construction and Analysis of Systems, pp. 337--340.
  Springer (2008)

\bibitem{dooley2016proving}
Dooley, M., Somenzi, F.: Proving parameterized systems safe by generalizing
  clausal proofs of small instances. In: International Conference on Computer
  Aided Verification. pp. 292--309. Springer (2016)

\bibitem{dutertre2014yices}
Dutertre, B.: Yices 2.2. In: Biere, A., Bloem, R. (eds.) Computer Aided
  Verification. pp. 737--744. Springer International Publishing, Cham (2014)

\bibitem{een2011efficient}
Een, N., Mishchenko, A., Brayton, R.: {Efficient Implementation of Property
  Directed Reachability}. In: Formal Methods in Computer Aided Design
  (FMCAD'11). pp. 125 -- 134 (Oct 2011)

\bibitem{een2003extensible}
E{\'e}n, N., S{\"o}rensson, N.: {An Extensible SAT-solver}. In: International
  conference on theory and applications of satisfiability testing. pp.
  502--518. Springer (2003)

\bibitem{emerson1996symmetry}
Emerson, E.A., Sistla, A.P.: Symmetry and model checking. Formal methods in
  system design  \textbf{9}(1-2),  105--131 (1996)

\bibitem{DBLP:journals/corr/abs-2008-09909}
Feldman, Y.M.Y., Sagiv, M., Shoham, S., Wilcox, J.R.: Learning the boundary of
  inductive invariants. CoRR  \textbf{abs/2008.09909} (2020),
  \url{https://arxiv.org/abs/2008.09909}

\bibitem{feldman2019complexity}
Feldman, Y.M., Immerman, N., Sagiv, M., Shoham, S.: Complexity and information
  in invariant inference. Proceedings of the ACM on Programming Languages
  \textbf{4}(POPL),  1--29 (2019)

\bibitem{feldman2019inferring}
Feldman, Y.M., Wilcox, J.R., Shoham, S., Sagiv, M.: Inferring inductive
  invariants from phase structures. In: International Conference on Computer
  Aided Verification. pp. 405--425. Springer (2019)

\bibitem{Fraleigh00}
Fraleigh, J.B.: A First Course in Abstract Algebra. Addison Wesley Longman,
  Reading, Massachusetts, 6th edn. (2000)

\bibitem{gario2015pysmt}
Gario, M., Micheli, A.: Pysmt: a solver-agnostic library for fast prototyping
  of smt-based algorithms. In: SMT workshop. vol.~2015 (2015)

\bibitem{german1992reasoning}
German, S.M., Sistla, A.P.: Reasoning about systems with many processes.
  Journal of the ACM (JACM)  \textbf{39}(3),  675--735 (1992)

\bibitem{v2019pretend}
v.~Gleissenthall, K., K{\i}c{\i}, R.G., Bakst, A., Stefan, D., Jhala, R.:
  Pretend synchrony: synchronous verification of asynchronous distributed
  programs. Proceedings of the ACM on Programming Languages  \textbf{3}(POPL),
  1--30 (2019)

\bibitem{godefroid1999exploiting}
Godefroid, P.: Exploiting symmetry when model-checking software. In: Formal
  Methods for Protocol Engineering and Distributed Systems, pp. 257--275.
  Springer (1999)

\bibitem{goel2019model}
Goel, A., Sakallah, K.: Model checking of verilog rtl using ic3 with
  syntax-guided abstraction. In: NASA Formal Methods Symposium. pp. 166--185.
  Springer (2019)

\bibitem{goel2020avr}
Goel, A., Sakallah, K.: Avr: Abstractly verifying reachability. In:
  International Conference on Tools and Algorithms for the Construction and
  Analysis of Systems. pp. 413--422. Springer (2020)

\bibitem{goel2019empirical}
Goel, A., Sakallah, K.A.: {Empirical Evaluation of IC3-Based Model Checking
  Techniques on Verilog RTL Designs}. In: Proc. of the Design, Automation and
  Test in Europe Conference (DATE). pp. 618--621. Florence, Italy (March 2019)

\bibitem{gurfinkel2018quantifiers}
Gurfinkel, A., Shoham, S., Vizel, Y.: Quantifiers on demand. In: International
  Symposium on Automated Technology for Verification and Analysis. pp.
  248--266. Springer (2018)

\bibitem{hawblitzel2015ironfleet}
Hawblitzel, C., Howell, J., Kapritsos, M., Lorch, J.R., Parno, B., Roberts,
  M.L., Setty, S., Zill, B.: Ironfleet: proving practical distributed systems
  correct. In: Proceedings of the 25th Symposium on Operating Systems
  Principles. pp. 1--17. ACM (2015)

\bibitem{hoenicke2017thread}
Hoenicke, J., Majumdar, R., Podelski, A.: Thread modularity at many levels: a
  pearl in compositional verification. ACM SIGPLAN Notices  \textbf{52}(1),
  473--485 (2017)

\bibitem{10.1145/3022187}
Karbyshev, A., Bj\o{}rner, N., Itzhaky, S., Rinetzky, N., Shoham, S.:
  Property-directed inference of universal invariants or proving their absence.
  J. ACM  \textbf{64}(1) (Mar 2017). \doi{10.1145/3022187},
  \url{https://doi.org/10.1145/3022187}

\bibitem{karbyshev2017property}
Karbyshev, A., Bj{\o}rner, N., Itzhaky, S., Rinetzky, N., Shoham, S.:
  Property-directed inference of universal invariants or proving their absence.
  Journal of the ACM (JACM)  \textbf{64}(1),  1--33 (2017)

\bibitem{pldi20folic3}
Koenig, J.R., Padon, O., Immerman, N., Aiken, A.: First-order quantified
  separators. In: Proceedings of the 41st ACM SIGPLAN Conference on Programming
  Language Design and Implementation. p. 703–717. PLDI 2020, Association for
  Computing Machinery, New York, NY, USA (2020). \doi{10.1145/3385412.3386018},
  \url{https://doi.org/10.1145/3385412.3386018}

\bibitem{kurshan1989structural}
Kurshan, R.P., McMillan, K.: A structural induction theorem for processes. In:
  Proceedings of the eighth annual ACM Symposium on Principles of distributed
  computing. pp. 239--247 (1989)

\bibitem{lamport1977proving}
Lamport, L.: Proving the correctness of multiprocess programs. IEEE
  transactions on software engineering (2),  125--143 (1977)

\bibitem{lamport2002specifying}
Lamport, L.: Specifying systems: the TLA+ language and tools for hardware and
  software engineers. Addison-Wesley Longman Publishing Co., Inc. (2002)

\bibitem{lamport2019part}
Lamport, L.: The part-time parliament. In: Concurrency: the Works of Leslie
  Lamport, pp. 277--317 (2019)

\bibitem{lamport2001paxos}
Lamport, L., et~al.: Paxos made simple. ACM Sigact News  \textbf{32}(4),
  18--25 (2001)

\bibitem{li2015paraverifier}
Li, Y., Pang, J., Lv, Y., Fan, D., Cao, S., Duan, K.: Paraverifier: An
  automatic framework for proving parameterized cache coherence protocols. In:
  International Symposium on Automated Technology for Verification and
  Analysis. pp. 207--213. Springer (2015)

\bibitem{ma2019i4}
Ma, H., Goel, A., Jeannin, J.B., Kapritsos, M., Kasikci, B., Sakallah, K.A.:
  I4: Incremental inference of inductive invariants for verification of
  distributed protocols. In: Proceedings of the 27th Symposium on Operating
  Systems Principles. ACM (2019)

\bibitem{ma2019towards}
Ma, H., Goel, A., Jeannin, J.B., Kapritsos, M., Kasikci, B., Sakallah, K.A.:
  Towards automatic inference of inductive invariants. In: Proceedings of the
  Workshop on Hot Topics in Operating Systems. pp. 30--36. ACM (2019)

\bibitem{marques1999grasp}
Marques-Silva, J.P., Sakallah, K.A.: Grasp: A search algorithm for
  propositional satisfiability. IEEE Transactions on Computers  \textbf{48}(5),
   506--521 (1999)

\bibitem{mcmillan1993symbolic}
McMillan, K.L.: {Symbolic Model Checking}. Kluwer Academic Publishers, Norwell,
  MA, USA (1993)

\bibitem{moskewicz2001chaff}
Moskewicz, M.W., Madigan, C.F., Zhao, Y., Zhang, L., Malik, S.: {Chaff:
  Engineering an Efficient SAT Solver}. In: DAC. pp. 530--535 (2001)

\bibitem{namjoshi2007symmetry}
Namjoshi, K.S.: Symmetry and completeness in the analysis of parameterized
  systems. In: International Workshop on Verification, Model Checking, and
  Abstract Interpretation. pp. 299--313. Springer (2007)

\bibitem{newcombe2015amazon}
Newcombe, C., Rath, T., Zhang, F., Munteanu, B., Brooker, M., Deardeuff, M.:
  How amazon web services uses formal methods. Communications of the ACM
  \textbf{58}(4),  66--73 (2015)

\bibitem{ip1996better}
Norris~IP, C., Dill, D.L.: Better verification through symmetry. Formal Methods
  in System Design  \textbf{9}(1),  41--75 (Aug 1996).
  \doi{10.1007/BF00625968}, \url{https://doi.org/10.1007/BF00625968}

\bibitem{owicki1976verifying}
Owicki, S., Gries, D.: Verifying properties of parallel programs: An axiomatic
  approach. Communications of the ACM  \textbf{19}(5),  279--285 (1976)

\bibitem{owre1992pvs}
Owre, S., Rushby, J.M., Shankar, N.: Pvs: A prototype verification system. In:
  International Conference on Automated Deduction. pp. 748--752. Springer
  (1992)

\bibitem{padon2016ivy}
Padon, O., McMillan, K.L., Panda, A., Sagiv, M., Shoham, S.: Ivy: Safety
  verification by interactive generalization. In: Proceedings of the 37th ACM
  SIGPLAN Conference on Programming Language Design and Implementation. pp.
  614--630. PLDI '16, ACM, New York, NY, USA (2016).
  \doi{10.1145/2908080.2908118},
  \url{http://doi.acm.org/10.1145/2908080.2908118}

\bibitem{piskac2010deciding}
Piskac, R., de~Moura, L., Bj{\o}rner, N.: Deciding effectively propositional
  logic using dpll and substitution sets. Journal of Automated Reasoning
  \textbf{44}(4),  401--424 (Apr 2010). \doi{10.1007/s10817-009-9161-6},
  \url{https://doi.org/10.1007/s10817-009-9161-6}

\bibitem{pnueli2001automatic}
Pnueli, A., Ruah, S., Zuck, L.: Automatic deductive verification with invisible
  invariants. In: International Conference on Tools and Algorithms for the
  Construction and Analysis of Systems. pp. 82--97. Springer (2001)

\bibitem{pong1995new}
Pong, F., Dubois, M.: A new approach for the verification of cache coherence
  protocols. IEEE Transactions on Parallel and Distributed Systems
  \textbf{6}(8),  773--787 (1995)

\bibitem{ghilardi2010backward}
Ranise, S., Ghilardi, S.: Backward reachability of array-based systems by smt
  solving: Termination and invariant synthesis. Logical Methods in Computer
  Science  \textbf{6} (2010)

\bibitem{sistla2000smc}
Sistla, A.P., Gyuris, V., Emerson, E.A.: Smc: a symmetry-based model checker
  for verification of safety and liveness properties. ACM Transactions on
  Software Engineering and Methodology (TOSEM)  \textbf{9}(2),  133--166 (2000)

\bibitem{Wilcox2015Verdi}
Wilcox, J.R., Woos, D., Panchekha, P., Tatlock, Z., Wang, X., Ernst, M.D.,
  Anderson, T.: Verdi: A framework for implementing and formally verifying
  distributed systems. In: Proceedings of the 36th ACM SIGPLAN Conference on
  Programming Language Design and Implementation. pp. 357--368. PLDI '15, ACM,
  New York, NY, USA (2015). \doi{10.1145/2737924.2737958},
  \url{http://doi.acm.org/10.1145/2737924.2737958}

\bibitem{zuck2004model}
Zuck, L., Pnueli, A.: Model checking and abstraction to the aid of
  parameterized systems (a survey). Computer Languages, Systems \& Structures
  \textbf{30}(3-4),  139--169 (2004)

\end{thebibliography}

\begin{subappendices}
\renewcommand{\thesection}{\appendixname~\Alph{section}}

\clearpage
\onecolumn
\section*{Appendices}
\vspace{10pt}
\noindent We include additional/supplementary material in the appendices, as follows:
\vspace{10pt}
\begin{itemize}
    \item[] \ref{app:algo}: \textit{\icpo Pseudo Code (detailed)}
    \begin{itemize}
    \item[--] Presents the detailed pseudo code of \icpo and \symic
    \end{itemize}
    \vspace{10pt}
    
    \item[] \ref{sec:proof}: \textit{Proof of Correctness }
    \begin{itemize}
    \item[--] Provides a correctness proof for symmetry-aware clause boosting during incremental induction (Section~\ref{sec:symclause}), and a correctness proof for quantifier inference (Section~\ref{sec:quantinfer})
    \end{itemize}
    \vspace{10pt}
    
    \item[] \ref{sec:opt}: \textit{Simple Enhancements to the \symic Algorithm} 
    \begin{itemize}
    \item[--] Describes simple enhancements to \symic learning as briefly mentioned in Section~\ref{sec:ic3po_highlevel}
    \end{itemize}
    \vspace{10pt}
    
    \item[] \ref{sec:eval_symboost}: \textit{Effect of Symmetry Learning in Incremental Induction }
    \begin{itemize}
    \item[--] Evaluates the effect of symmetry-aware learning in finite-domain incremental induction with a detailed comparison between \icpo and \ifour
    \end{itemize}
    \vspace{10pt}
    
    \item[] \ref{sec:eval_statistical}: \textit{Statistical Analysis with Multiple SMT Solver Seeds}
    \begin{itemize}
    \item[--] Provides a statistical analysis of the experiments from Section~\ref{sec:evaluation} through multiple runs for each tool with different solver seeds
    \end{itemize}
    \vspace{10pt}
    
    \item[] \ref{sec:eval_human}: \textit{Comparison against Human-Written Invariants}
    \begin{itemize}
    \item[--] Compares \icpo's automatically-generated quantified inductive invariants against human-written invariant proofs on several metrics
    \end{itemize}
    \vspace{10pt}
    
    \item[] \ref{sec:eval_order}: \textit{Ordered Domains, Ring Topology, and Special Variables}
    \begin{itemize}
    \item[--] Describes an extension to \icpo that allows handling totally-ordered domains, as well as further details relating to ring topology and special variables, along with a preliminary evaluation
    \end{itemize}
    \vspace{10pt}
    
    \item[] \ref{app:sizes}: \textit{Finite Instance Sizes used in the Experiments}
    \begin{itemize}
    \item[--] Lists down the instance sizes for \icpo and \ifour for each protocol in the evaluation (Section~\ref{sec:evaluation})
    \end{itemize}
    \vspace{10pt}
    
\end{itemize}

\clearpage
\onecolumn
\section{\icpo~Pseudo Code (detailed)}
\label{app:algo}
This section presents the detailed pseudo code of \icpo and \symic.

\setlength{\textfloatsep}{15pt}
\begin{algorithm}[H]
\begin{algorithmic}[1]
\algrenewcommand\alglinenumber[1]{\fontsize{7}{8} \selectfont #1}
\tlasize
\Procedure{\regularf{\icpo}}{$\pro$, $\sz_0$} \Comment{$\pro \triangleq [S,R,Init,T,P]$, and $\sz_0$ is the initial base size}
\State $reuse$ $\gets$ $\{\}$
\State $\szk$ $\gets$ $\sz_0$
\State $\invk, \cexk$ $\gets$ \regularf{\symic}($\fprok$, $reuse$) \Comment{run symmetric incremental induction on $\fprok \triangleq \pro(\szk)$} \label{line:symic3}	
\If {$\cexk$ is not empty} \Comment{counterexample found}
	\State \Return \texttt{Violated}, $\cexk$ \Comment{property is violated}
\Else	\Comment{property proved for the finite protocol instance $\fprok$}
    \ForEach {$\s_{\fin{i}} \in S$}
        \If {\Not~\regularf{IsInductiveInvariantFinite}($\invk$, $\fpronext$)}	
            \State $reuse$ $\gets$ \{ $\Phi~|~\Phi \in \invk$ and $Init \to \Phi$ and $Init \wedge T \to \Phi'$ in $\fpronext$ \}
            \State $\szk$ $\gets$ $\sznext$ \Comment{failed convergence checks for sort $\s_{\fin{i}}$, increase instance size}
            \State \Goto{line:symic3} \Comment{re-run \symic~with the increased size}
    	\EndIf
    \EndFor
    \If {\Not~\regularf{IsInductiveInvariantUnbounded}($\invk$, $\pro$)}
        \State \texttt{< never occurred >} \Comment{unbounded check failed}
        \State \Return \texttt{Error}, \texttt{Increase $\sz_0$} 
	\EndIf
	\State \Return \texttt{Safe}, $\invk$ \Comment{property is proved safe with proof certificate $\invk$}
\EndIf
\EndProcedure

\captionof{algorithm}{\textit{\textbf{IC3} for \textbf{P}r\textbf{o}ving \textbf{P}r\textbf{o}tocol \textbf{P}r\textbf{o}perties}}
\label{alg:ic3po}
\end{algorithmic}
\end{algorithm}

Algorithm~\ref{alg:ic3po} presents the detailed pseudo code of \icpo. Let $\sz: S \rightarrow \mathbb{N}$ be a function that maps each sort $\fin{s_i} \in S$ to a sort size $\fin{|s_i|}$. Given a protocol specification $\pro$ and an initial base size $\sz_0$, \icpo~invokes \symic~on the finite protocol instance $\fprok \triangleq \pro(\szk)$, where $\szk$ is initialized to $\sz_0$ (lines 2-4). Upon termination, \symic~either a) produces a quantified inductive invariant $\invk$ that proves the property for $\fprok$, or b) a counterexample trace $\cexk$ that serves as a finite witness to its violation in both $\fprok$ and the unbounded protocol $\pro$ (lines 4-6). If the property holds for $\fprok$, \icpo~performs finite convergence checks (Section~\ref{sec:invcheck}) to check whether or not the invariant extends beyond $\fprok$ (lines 8-12), by checking whether or not $\invk$ is an inductive invariant for the larger finite instance $\fprok^{\fin{i}} \triangleq \fpronext$ for each $\s_{\fin{i}} \in S$, where $\sznext \triangleq [\szk$ \EXCEPT \bang $[\s_{\fin{i}}] = \szk(\s_{\fin{i}})+1]$. If all finite checks pass, $\invk$ is checked whether an inductive invariant in the unbounded domain (lines 13-15) using the standard induction checks-- a) $Init \to \invk$, and b) $\invk \wedge T \to \invk'$ in the unbounded domain. If all these checks pass, \icpo~emits the unbounded invariant $\invk$, that holds for the unbounded $\pro$ and is a proof certificate for the safety property (line 16). Otherwise, it re-starts \symic~on a finite instance with an increased size $\sznext$ (lines 11-12), while seeding in all the strengthening assertions in $\invk$ that are safe to learn in the first frame for the new \symic iteration (line 10).

\begin{algorithm}[!t]
\algrenewcommand\alglinenumber[1]{\fontsize{7}{8} \selectfont #1}
\caption{\textit{Symmetric Incremental Induction}}
\begin{algorithmic}[1]
\setcounter{ALG@line}{16}
\tlasize
\Procedure{\regularf{\symic}($\fprok$, $reuse$)}{} \Comment{$\fpro \triangleq [S,R,\hat{Init},\hat{T},\hat{P}]$}
\Statex \Comment{$reuse$ is a set of seed assertions that are safe to learn in the frame $F_1$}
\State $F$ $\gets$ $\emptyset$, $Cex$ $\gets$ $\emptyset$ \Comment{$\fprok$, $F$, $Cex$ are global data structures}
\If {\texttt{ SAT ? [ $\hat{Init} \wedge \neg \hat{P}$ ]: model $\fin{m}$}} \Comment{initial states check}
    \State $state$ $\gets$ \regularf{StateAsCube}($\fin{m}$) \Comment{get a single state from model $\fin{m}$, in cube form}
	\State $Cex.extend(state)$ \Comment{property is trivially violated}
	\State \Return $\emptyset$, $Cex$ \Comment{return the counterexample}
\EndIf
\State $F.extend(\hat{Init})$ \Comment{setup the initial frame}

\While{$\top$}
    \State $N$ $\gets$ $F.size() - 1$
    \If {\texttt{ SAT ? [ $F_N \wedge \hat{T} \wedge \neg \hat{P}'$ ]: model $\fin{m}$}}  
    \Statex \Comment{check the topmost frame for counterexample-to-induction (CTI)}
    \State $state$ $\gets$ \regularf{StateAsCube}($\fin{m}$) \Comment{found a CTI}
	    \If {\regularf{SymRecBlockCube}($state$, $N$)} \Comment{try recursively blocking the CTI}
        	\State \Return $\emptyset$, $Cex$ \Comment{failed to block CTI, return the counterexample}
        \EndIf
    \Else \Comment{no CTI in the topmost frame}
    	\State $F.extend(\hat{P})$ \Comment{add a new frame}
	    \If {$N = 0$} \Comment{add reusable seed assertions to the frame $F_1$}
        	\State $F[1].add(reuse)$
        \EndIf
    	\If {\regularf{ForwardPropagate()}} \Comment{propagate inductive assertions forward}
    		\State \Return $F_{converged}$, $\emptyset$ 
    		\Statex \Comment{frames converged, return $F_{converged}$ as the inductive invariant}
    	\EndIf
	\EndIf
\EndWhile

\EndProcedure



\Statex
\Procedure{\regularf{SymRecBlockCube}}{$cti, i$} \Comment{$cti$ can reach $\neg \hat{P}$ in $F.size() - i$ steps}
\State $Cex.extend(cti)$ \Comment{add the CTI to the counterexample}
\If {$i = 0$} \Comment{check if reached the initial states}
    \State \Return $\top$ \Comment{reached initial states, property is violated}
\EndIf
\If {\texttt{ SAT ? [ $F_{i-1} \wedge \hat{T} \wedge cti'$ ]: model $\fin{m}$}} 
\Statex \Comment{check if $cti$ is reachable from previous frame}
    \State $state$ $\gets$ \regularf{StateAsCube}($\fin{m}$) 
    \Statex \Comment{$state$ is the new CTI reachable to $\neg \hat{P}$ in $(F.size() - i) + 1$ steps}
    \State \Return \regularf{SymRecBlockCube}($state$, $i-1$)  \Comment{try blocking the new CTI}
\Else \Comment{$cti$ is unreachable from the previous frame}
	\State $uc'$ $\gets$ \regularf{MinimalUnsatCore}($F_{i-1} \wedge \hat{T}$, $cti'$) \Comment{get MUS from UNSAT query}
	\State $\varphi$ $\gets$ $\neg uc$ \Comment{negate $uc$ to get the quantifier-free clause}
	\State $\quant$ $\gets$ \regularf{\symboost}($\varphi$) \Comment{symmetry-aware clause boosting with quantifier inference}
	\State $\quant$ $\gets$ \regularf{AntecedentReduction}($\quant$, $i$) \Comment{antecedent reduction (optional),~\ref{sec:otheropt}}	
	\State $\quant$ $\gets$ \regularf{EprReduction}($\quant$, $i$) \Comment{EPR reduction (optional),~\ref{sec:epr}}
	\State \texttt{Learn($\quant$, $F_i$)} \Comment{learn $\quant$ in frame $i$}
	\State \Return $\perp$
\EndIf
\State
\EndProcedure

\end{algorithmic}
\label{alg:symrec}
\end{algorithm}

Algorithm~\ref{alg:symrec} describes the symmetric incremental induction algorithm. The procedure first checks whether the property can be trivially violated (lines 19-22), and if not, starts recursively deriving and blocking counterexamples-to-induction (CTI) from the topmost frame (lines 24-35). Given a solver model $\fin{m}$, a state cube is derived as a single state represented as a cube, i.e., a conjunction of literals assigning each state variable with a value based on its assignment in $\fin{m}$ (lines 20, 27, 41). Lines 32-33 add the seed assertions in the given $reuse$ set to the first frame $F_1$. 
\symic differs from the standard IC3 algorithm majorly in symmetry-aware quantified learning (line 46) and simple enhancements (lines 47-48).

\begin{algorithm}[!t]
\begin{algorithmic}[1]
\algrenewcommand\alglinenumber[1]{\fontsize{7}{8} \selectfont #1}
\setcounter{ALG@line}{51}
\tlasize
\Procedure{\symboost}{$\varphi$} \Comment{$\varphi$ is the quantifier-free clause}
\State $V_{\forall}$ $\gets$ \{\}, $V_{\exists}$ $\gets$ \{\} \Comment{a set of universally/existential quantified variables}
\State $body$ $\gets$ $\varphi$ \Comment{starting with $\varphi$, $body$ is recursively generated}
\Statex \Comment{$V_{\forall}$, $V_{\exists}$ and $body$ are global data structures}
\ForEach {sort $\s$ that appears in clause $\varphi$}
    \State $\pi(\varphi, \s)$ $\gets$ \regularf{PartitionDistribution}($\varphi$, $\s$) 
    \Statex \Comment{create a partition on constants in $\s$ based on their occurrence in $\varphi$}
    \If {$\numc(\varphi, \s) < |\s|$}
        \State $(V_{\forall}, V_{\exists}, body)$ $\gets$ \regularf{Infer$\forall$}($\varphi$, $\pi(\varphi, \s)$) \Comment{infer $\forall$ for sort $\s$, refer \S\ref{sec:quantinfer_main}.A}
    \ElsIf {$|\pi(\varphi, \s)| = 1$} \Comment{partition $\pi(\varphi, \s)$ contains a single cell}
        \State $(V_{\forall}, V_{\exists}, body)$ $\gets$ \regularf{Infer$\exists$}($\varphi$, $\pi(\varphi, \s)$) \Comment{infer $\exists$ for sort $\s$, refer \S\ref{sec:quantinfer_main}.B.I}
    \ElsIf {\textit{all but a few scenario}} \Comment{partition $\pi(\varphi, \s)$ contains multiple cells}
        \State $(V_{\forall}, V_{\exists}, body)$ $\gets$ \regularf{Infer$\forall\exists$}($\varphi$, $\pi(\varphi, \s)$)  \Comment{infer $\forall\exists$ for sort $\s$, refer \S\ref{sec:quantinfer_main}.B.II}
    \Else
        \State \texttt{< never occurred >}
        \Statex \Comment{infer $\forall$ by default (may not be compact, though correct for the current instance)}
        \State $(V_{\forall}, V_{\exists}, body)$ $\gets$ \regularf{Infer$\forall$}($\varphi$, $\pi(\varphi, \s)$)
    \EndIf
\EndFor
\State $\quant$ $\gets$ $\forall V_{\forall}.~\exists V_{\exists}.~body$ \Comment{stitch quantifiers for different sorts as $\forall_{\dots}~\exists_{\dots}~<body>$}
\State \Return $\quant$ \Comment{$\Phi$ is the quantified predicate to learn in a \symic~frame}
\EndProcedure
\end{algorithmic}
\captionof{algorithm}{\textit{Symmetry-aware Clause Boosting with Quantifier Inference}}
\label{alg:symgen}
\end{algorithm}

The core of the \symic~algorithm is the \symboost~algorithm, presented in Algorithm~\ref{alg:symgen}.
\symboost~is a simple and extendable procedure to perform symmetry-aware clause boosting and quantifier inference, as explained in detail in Sections~\ref{sec:symclause} and~\ref{sec:quantinfer}.
Starting from a given quantifier-free clause $\varphi$, the algorithm constructs a symmetrically-boosted quantified predicate $\quant$ (line 67) by iteratively inferring quantifiers for each sort $\s$ (lines 55-65), and stitching them together (line 66).
The algorithm maintains a set of universal and existential variables (line 53) and a $body$ (line 54), that are iteratively modified based on the quantifier inference for each sort.
For each sort $\s$, the algorithm first generates $\pi(\varphi, \s)$ (line 56) based on how constants in sort $\s$ appear in the literals of $\varphi$ (whether identically or not). The next step is to infer quantifiers using $\numc(\varphi, \s)$ and $\pi(\varphi, \s)$ (lines 57-65): a) infer universal quantifiers when $\numc(\varphi, \s) < |\s|$, b) otherwise if all constants of $\s$ appear in $\varphi$ identically, infer existential quantifier, c) otherwise if \textit{all but a few scenario}, infer $\forall\exists$ based on the partitioning of constants in $\pi(\varphi, s)$, and d) otherwise, infer $\forall$ by default (this case has not occurred).
Changing the iteration order in line 55 doesn't result in any difference, and is ensured during the recursive building of the $body$. At the end, a single quantified predicate $\Phi$ is derived by stitching together the quantified variables in $V_{\forall}$ and $V_{\exists}$ with the $body$ as $\forall_{\dots}~\exists_{\dots}~<body>$ (line 66).

\clearpage
\onecolumn
\section{Proof of Correctness}
\label{sec:proof}
\vspace{10pt}
\subsection{Correctness Proof for Symmetric Incremental Induction}
\label{sec:proof_symboost}
This section provides a correctness proof for symmetry-aware clause boosting during incremental induction (Section~\ref{sec:symclause}).
\vspace{10pt}

Like the invariance of $\hat{Init}, \hat{T}$, and $\hat{P}$ under any permutation $\gamma \in G$ (refer~(\ref{eqn:ITP_invariance})), the logical orbit of a clause $\varphi$ is also invariant under such permutations, i.e.,
{\small
$${\left[ {{\varphi ^{L(G)}}} \right]^\gamma } \leftrightarrow {\varphi ^{L(G)}}$$
}%

\vspace{10pt}
\begin{lemma}
\label{lemma:frame}
For any \symic~frame $F_{i}$, $F_{i}^{\gamma} \same F_{i}$ for any $\gamma \in G$.
\end{lemma}
\begin{proof}
Recall that $\hat{Init}^{\gamma} \same \hat{Init}$ and $\hat{P}^{\gamma} \same \hat{P}$. 
The condition $F_{i}^{\gamma} \same F_{i}$ is trivially true for $i=0$ since $F_0 = \hat{Init}$. 
When $i > 0$, the condition is true during frame initialization since each frame is initialized to $\hat{P}$. 
When blocking a cube $\neg \varphi$ in $F_i$, incremental induction with symmetry boosting refines $F_i$ with the complete logical orbit $\varphi^{L(G)}$ of $\varphi$. Since $\left[ \varphi^{L(G)} \right]^{\gamma} \same \varphi^{L(G)}$, the logical invariance of $F_i$ under $\gamma$, continues to be preserved in all backward reachability updates. 
\qed
\end{proof}

The following theorem establishes the correctness of symmetry-aware clause boosting in incremental induction.
\vspace{10pt}
\begin{theorem}
\label{lemma:correct1}
If a quantifier-free cube $\neg \varphi$ is unreachable from frame $F_{i-1}$, i.e., $F_{i-1} \wedge \hat{T} \wedge \neg [\varphi]'$ is unsatisfiable, then $F_{i-1} \wedge \hat{T} \wedge \neg [{\varphi^{L(G)}}]'$ is also unsatisfiable.
\end{theorem}
\begin{proof}
Let $Q \triangleq F_{i-1} \wedge \hat{T} \wedge \neg [\varphi]'$ and assume that $Q$ is unsatisfiable. Consider any permutation $\gamma \in G$ and the corresponding permuted formula $Q^{\gamma} \triangleq F_{i-1}^{\gamma} \wedge \hat{T}^{\gamma} \wedge \neg \left[{\varphi^{\gamma}}\right]'$.
Since permuting the sort constants simply re-arranges the protocol's state variables in a formula without affecting its satisfiability, $Q$ and $Q^{\gamma}$ must be equisatisfiable, and hence  $Q^{\gamma}$ is unsatisfiable. 

Noting that $\hat{T}$ and $F_{i-1}$ are invariant under $\gamma \in G$ (from (\ref{eqn:ITP_invariance}) and Lemma~\ref{lemma:frame}), we obtain 
$Q^\gamma =F_{i-1} \wedge \hat{T} \wedge \neg \left[{\varphi^{\gamma}}\right]'$ proving that if cube $\neg \varphi$ is unreachable from frame $F_{i-1}$, then its image under any $\gamma \in G$ is also unreachable. Therefore, $F_{i-1} \wedge \hat{T} \wedge \neg [{\varphi^{L(G)}}]'$ is unsatisfiable.
\qed
\end{proof}

\subsection{Correctness Proof for Quantifier Inference}
\label{sec:proof_qi}
This section provides a correctness proof sketch for quantifier inference (Section~\ref{sec:quantinfer}).

\vspace{10pt}
\begin{theorem}
\label{lemma:correct2_new}
Given a finite instance $\fpro$, let $\varphi$ be such that $0 < \numc(\varphi,\s) < |\s|$ for some sort $\s \in S$. Let $\Phi(\s)$ be the quantified predicate obtained by applying \symic's quantifier inference for $\s$. $\Phi(\s)$ is logically equivalent to $\varphi^{L(Sym(\s))}$.
\end{theorem}
\begin{proof}
Let $\gamma$ be any permutation in $Sym(\s)$, and let $n \triangleq \numc(\varphi, \s)$.
Let $\phih$ be the clause obtained by replacing in $\varphi$ each constant $\fin{c_i} \in \s$ by a corresponding variable $V_i$ of sort $\s$.

Let $A \triangleq [~(V_1=\fin{c_1}) \wedge \dots \wedge (V_n=\fin{c_n})~] \to \widehat{\varphi}$.
By the transitivity of equality, $A \same \varphi$. 
Let $B \triangleq \bigwedge \limits_{\gamma  \in Sym(\s)} A^{\gamma}$. Since $A \same \varphi$, therefore, $B \same \varphi^{L(Sym(\s))}$, and can be re-written as:
\begin{small}
\begin{align}
B &= \bigwedge \limits_{\gamma  \in Sym(\s)} \big(~[~(V_1=\fin{c_1}) \wedge \dots \wedge (V_n=\fin{c_n})~] \to \widehat{\varphi}~\big)^{\gamma} \label{eq:conjunct_s2}\\
&= \bigwedge \limits_{\gamma  \in Sym(\s)} [~(V_1=\fin{c_1}) \wedge \dots \wedge (V_n=\fin{c_n})~]^{\gamma} \to \phih \label{eq:conjunct_s3}\\
&= \forall~V_1 \dots V_n.~(\text{distinct}~V_1 \dots V_n) \to \phih \label{eq:conjunct_s4}\\
&= \Phi(\s) \label{eq:conjunct_s5}
\end{align}
\end{small}

\noindent (\ref{eq:conjunct_s2}) \& (\ref{eq:conjunct_s3}) are equal since $\phih$ does not contain any constant of sort $\s$, and hence $\left[\phih\right]^{\gamma} \same \phih$.
\noindent (\ref{eq:conjunct_s3}) \& (\ref{eq:conjunct_s4}) are equal since the antecedents in (\ref{eq:conjunct_s3}) cover all possible assignments of variables $(V_1,\dots,V_n)$ to $n$ distinct constants of sort $\s$. There are total ${|\s| \choose n} \times n!$ possible assignments of the variables in (\ref{eq:conjunct_s4}) to $n$ distinct constants of sort $\s$, one each corresponding to the ${|\s| \choose n} \times n!$ permutations in $Sym(\s)$ that yield a logically-distinct antecedent in (\ref{eq:conjunct_s3}). 
\noindent (\ref{eq:conjunct_s4}) \& (\ref{eq:conjunct_s5}) are equal since given $\numc(\varphi, \s) < |\s|$.

\noindent Since $B \same \varphi^{L(Sym(\s))}$, therefore $\Phi(\s) \same \varphi^{L(Sym(\s))}$.
\qed

\end{proof}

\vspace{10pt}
\begin{theorem}
\label{lemma:correct2_new2}
Given a finite instance $\fpro$, let $\varphi$ be such that all constants of a sort $\s \in S$ appear identically in the literals of $\varphi$. Let $\Phi(\s)$ be the quantified predicate obtained by applying \symic's quantifier inference for $\s$. $\Phi(\s)$ is logically equivalent to $\varphi^{L(Sym(\s))}$.
\end{theorem}
\begin{proof}
Let $\gamma$ be any permutation in $Sym(\s)$.
Since given all constants in sort $\s$ appear identically in the literals of $\varphi$, therefore $\pi(\varphi, \s)$ consists of a single cell, and any permutation $\gamma \in Sym(\s)$ does not result in a new logically-distinct clause, i.e., $\varphi^{\gamma} \same \varphi$. As a result, $\varphi^{L(Sym(\s))} \same \varphi$.

\noindent Without loss of generality, $\varphi$ can be written as:
\begin{small}
\begin{align}
& \varphi = \varphi_{others} \vee \bigvee \limits_{\fin{c_i} \in \s} \varphi_\s(\fin{c_i}) \label{eq:exists_s0}
\end{align}
\end{small}
where $\varphi_{others}$ is the disjunction of literals in $\varphi$ that do not contain any constant of sort $\s$, and $\varphi_\s(\fin{c_i})$ is the disjunction of literals in $\varphi$ that contain a constant $\fin{c_i} \in \s$. Note that $\varphi_{others}$ can be $\perp$.

Let $\widehat{\varphi_\s}$ be the clause obtained by replacing in $\varphi_\s(\fin{c_i})$ each constant $\fin{c_i} \in \s$ by a variable $V$ of sort $\s$. Note that since all constants of sort $\s$ appear identically in the literals of $\varphi$, therefore $\widehat{\varphi_\s}$ is the same for each $\fin{c_i} \in \s$.
The clause $\varphi$ can therefore be re-written as:
\begin{small}
\begin{align}
\varphi 
&= \varphi_{others} \vee \bigvee \limits_{\fin{c_i} \in \s} (V = \fin{c_i}) \to \widehat{\varphi_\s} \label{eq:exists_s1} \\
&= \varphi_{others} \vee \exists~V.~~\widehat{\varphi_\s} \label{eq:exists_s2} \\
&= \Phi(\s) \label{eq:exists_s3}
\end{align}
\end{small}

\noindent (\ref{eq:exists_s0}) \& (\ref{eq:exists_s1}) are equal due to the transitivity of equality.
\noindent (\ref{eq:exists_s1}) \& (\ref{eq:exists_s2}) are equal since expanding the existential quantifier as a disjunction over all possible assignments of the variable $V$ gives the expression in (\ref{eq:exists_s1}).
\noindent (\ref{eq:exists_s2}) \& (\ref{eq:exists_s3}) are equal since $\numc(\varphi, \s) = |\s|$ and $|\pi(\varphi, \s) = 1|$, and hence \symic~infers $\Phi(\s)$ as (\ref{eq:exists_s2}).
\noindent Since $\varphi \same \varphi^{L(Sym(\s))}$, therefore $\Phi(\s) \same \varphi^{L(Sym(\s))}$.
\qed

\end{proof}

\clearpage
\onecolumn
\section{Simple Enhancements to the \icpo~Algorithm}
\label{sec:opt}
This section describes simple enhancements to \symic learning as mentioned in Section~\ref{sec:ic3po_highlevel}.

\subsection{Antecedent Reduction}
\label{sec:otheropt}
\noindent
Antecedent reduction strengthens a quantified predicate $\Phi$ by dropping the antecedent $(\text{distinct}~\dots)$ and checking the unsatisfiability of the query [~$F_{i-1} \wedge \hat{T} \wedge \neg \quant'$~]. For example, $\quant_2$ from (\ref{eq:ex_quant_clause2}) can possibly be strengthened by dropping $(\text{distinct}~X_1~X_2)$ from the antecedent to get $\quant_{new}$, if the query [~$F_{i-1} \wedge \hat{T} \wedge \neg \quant_{new}'$~] is unsatisfiable, where
{\small
\begin{equation*}
\Phi_{new} =~ \forall X_1, X_2 \in \fin{value}.~ \neg decision(X_1) \vee decision(X_2)
\end{equation*}
}%
If instead, the query is satisfiable, the original predicate $\Phi_2$ should be learnt.
    
\noindent
\subsection{EPR Reduction}
\label{sec:epr}
\noindent
With the quantifier inference employed by \symboost~(Algorithm~\ref{alg:symgen}), \symic~can produce predicates with alternating quantifiers, which can result in quantifier-alternation cycles. For example, our running example already includes a quantifier alternation from $\fin{quorum} \longrightarrow \fin{node}$ (Figure~\ref{fig:tla_naive}, line 3). 
Consider an example predicate:
{\small
\begin{equation}
\quant =~\forall Y \in \fin{node},~\exists Z \in \fin{quorum}.~ member(Y, Z) \nonumber
\end{equation}
}%
The quantified predicate $\Phi$ adds the arc $\fin{node} \longrightarrow \fin{quorum}$, generating a quantifier-alternation cycle: $$\fin{quorum} \longrightarrow \fin{node} \longrightarrow \fin{quorum}$$

Even though there are no undecidability concerns while reasoning over the finite instance $\fpro$ (since the sort domains are finite), it is desirable to avoid quantifier-alternation cycles and derive the invariant in the EPR fragment~\cite{piskac2010deciding} of FOL. Restricting to the EPR fragment allows robustly checking the inductive invariant over the unbounded protocol $\pro$. Note that \icpo~performs invariant construction as well as finite convergence checks both in a finite domain (as detailed in Section~\ref{sec:ic3po_highlevel}).

We can additionally strengthen the learning to be within the EPR fragment, by \textit{pushing out} existential quantifiers and avoid generation of quantifier-alternation cycle. For example, the EPR-reduced version $\quant_{epr}$ of $\quant$ is
{\small
\begin{equation}
\quant_{epr} =~\exists Z \in \fin{quorum},~\forall Y \in \fin{node}.~ member(Y, Z) \nonumber
\end{equation}
}%
If we consider both $\Phi$ and its negation $\neg \Phi$ (as needed during induction checks), EPR-reduction basically \textit{flips} the quantifier-alternation arcs. For example, the quantifier-alternation graph with the EPR-reduced predicate $\Phi_{epr}$ (instead of $\Phi$) is:
$$\fin{quorum} \longrightarrow \fin{node} \longleftarrow \fin{quorum}$$
$\neg \Phi_{epr}$ adds the arc $\fin{node} \longleftarrow \fin{quorum}$.

Logically, \textit{pushing out} the existential quantifier results in a \textit{reduced}/\textit{stricter} formula, with $\quant_{epr} \to \quant$, but $\quant \centernot\to \quant_{epr}$ (hence we call it EPR ``reduction''). Intuitively, this difference is analogous to the difference in the statements:
\begin{small}
\begin{equation*}
{Likes}_{\forall\exists} :=~\text{Everyone likes someone} \qquad {Likes}_{\exists\forall} :=~\text{Someone is liked by everyone}
\end{equation*}
\end{small}
where ${Likes}_{\exists\forall} \to {Likes}_{\forall\exists}$, but ${Likes}_{\forall\exists} \centernot\to {Likes}_{\exists\forall}$.

We can add EPR reduction in the incremental induction procedure with \symic, that enables learning the EPR-reduced form $\quant_{epr}$ instead of $\quant$ \textit{only when it is safe}, i.e., only when $\neg \quant_{epr}$ is still unreachable from the previous incremental induction frame $F_{i-1}$. We do so by checking the unsatisfiability of the finite domain (and hence decidable) query [~$F_{i-1} \wedge \hat{T} \wedge \neg \quant_{epr}'$~]. If the query is unsatisfiable, we learn the strengthened EPR-reduced predicate $\Phi_{epr}$. Else, the original form, i.e., $\Phi$, is learnt.

\vspace{20pt}
\noindent Note- Both simple enhancements presented in this section were left disabled in \icpo for all experiments in this paper to focus the evaluation on the main paper contents. Initial investigation with these enhancements shows significant benefits in performance and robustness, with hardly any overhead.

\clearpage
\onecolumn
\section{Effect of Symmetry Learning in Incremental Induction}
\label{sec:eval_symboost}
\noindent
This section evaluates the effect of symmetry-aware clause boosting in finite-domain incremental induction with a detailed comparison between \icpo and \ifour.

Table~\ref{tab:table5} compares the effect of symmetry-aware learning in incremental induction for the problems solved by both \icpo~and \ifour. The table compares the number of SMT solver calls made and counterexamples-to-induction (CTI) encountered during the incremental induction procedure, as well as the number of assertions in the final (quantified) inductive invariant. \symic's symmetry boosting helps \icpo~to make orders of magnitude fewer SMT solver calls compared to \ifour and solve the problem after discovering many fewer CTIs.

Overall, Table~\ref{tab:table5} justifies the runtime speedups observed in Table~\ref{tab:table1}, and confirms the benefits of symmetry-aware learning.

\begin{table}[H]
\small
\setlength\tabcolsep{3pt}

\begin{center}
\resizebox{\textwidth}{!}{
\begin{tabular}{l|r|r|r|r|r|r}
        \multicolumn{1}{c}{} & \multicolumn{3}{c|}{\icpo} & \multicolumn{3}{c}{\ifour} \\
        \hline
        \multicolumn{1}{c|}{Protocol (\#$13$)} & \multicolumn{1}{c|}{~~~\#SMT~~~} & \multicolumn{1}{c|}{~~\#CTI~~} & \multicolumn{1}{c|}{~~\#Inv~~} & \multicolumn{1}{c|}{~~~\#SMT~~~} & \multicolumn{1}{c|}{~~\#CTI~~} & \multicolumn{1}{c}{~~\#Inv~~} \\
	\hline
	tla-consensus $\hfill$  & 13 & 0 & 1 & 7 & 0 & 1 \\
	i4-lock-server $\hfill$  & 31 & 1 & 2 & 35 & 2 & 2 \\
	ex-quorum-leader-election $\hfill$  & 117 & 7 & 5 & 15429 & 847 & 14 \\
	tla-simple $\hfill$  & 273 & 23 & 3 & 1319 & 41 & 3 \\
	ex-lockserv-automaton $\hfill$  & 568 & 51 & 12 & 1731 & 156 & 15 \\
	pyv-sharded-kv $\hfill$  & 572 & 25 & 8 & 2101 & 170 & 15 \\
	pyv-lockserv $\hfill$  & 676 & 58 & 12 & 1606 & 142 & 15 \\
	i4-learning-switch $\hfill$  & 567 & 32 & 9 & 26345 & 1310 & 11 \\
	ex-simple-decentralized-lock $\hfill$  & 2155 & 87 & 15 & 5561 & 490 & 22 \\
	i4-two-phase-commit $\hfill$  & 2131 & 68 & 11 & 4045 & 288 & 16 \\
	pyv-consensus-wo-decide $\hfill$  & 1866 & 141 & 9 & 41137 & 2451 & 42 \\
	pyv-consensus-forall $\hfill$  & 3423 & 247 & 10 & 156838 & 10316 & 44 \\
	pyv-learning-switch $\hfill$  & 3352 & 112 & 13 & 51021 & 3639 & 49 \\
	\hline
	\multicolumn{1}{l}{$\sum$ \#SMT} & \multicolumn{1}{|c}{15744} & \multicolumn{2}{l}{(\textbf{19.5x} better)} & \multicolumn{3}{|c}{307175}\\
	\multicolumn{1}{l}{$\sum$ \#CTI} & \multicolumn{1}{|c}{852} & \multicolumn{2}{l}{(\textbf{23.3x} better)} & \multicolumn{3}{|c}{19852}\\
	\multicolumn{1}{l}{$\sum$ \#Inv} & \multicolumn{1}{|c}{110} & \multicolumn{2}{l}{(\textbf{2.3x} better)} & \multicolumn{3}{|c}{249}\\
	\hline
\end{tabular}
}
\captionsetup{justification=centering}
\caption{
Comparison of different incremental induction metrics between \icpo~and \ifour for the problems solved by both  \\
\#SMT: number of solver queries, \#CTI: number of counterexamples-to-induction \\
\#Inv: number of assertions in the final (quantified) inductive invariant
}
\label{tab:table5}
\end{center}
\vspace{-10pt}
\end{table}

\clearpage
\onecolumn
\section{Statistical Analysis with Multiple SMT Solver Seeds}
\label{sec:eval_statistical}
\noindent
This section provides a statistical analysis of the experiments from Section~\ref{sec:evaluation} through multiple runs for each tool with different solver seeds.

Different tools perform best with different SMT solvers (e.g., \ifour uses a combination of Yices 2~\cite{dutertre2014yices} and Z3~\cite{demoura2008z3}, \folic uses Z3 and CVC4~\cite{BCD_11}, while \updr and \icpo~use Z3).\footnote{We used Yices 2 version 2.6.2, Z3 version 4.8.9 and CVC4 version 1.7.} 
For the results presented in Table~\ref{tab:table1}, a fixed SMT solver seed (i.e., $seed = 1$) was used for all tools. To get an idea of the effect of randomness in SMT solving, we performed $10$ runs with different solver seeds for each tool on all protocols, and compared the runtime mean and standard deviation.


\begin{table}[H]
\small
\setlength\tabcolsep{3pt}
\begin{center}
\resizebox{\textwidth}{!}{
\begin{tabular}{l|rrr|rrr|rrr|rrr}
         \multicolumn{1}{c}{} & \multicolumn{3}{c}{\icpo} & \multicolumn{3}{c}{\ifour} & \multicolumn{3}{c}{\updr} & \multicolumn{3}{c}{\folic} \\
        \hline
        \multicolumn{1}{c}{Protocol (\#29)} & \multicolumn{1}{|c}{\#} & \multicolumn{1}{c}{Time} & \multicolumn{1}{c}{$\sigma$} & \multicolumn{1}{|c}{\#} & \multicolumn{1}{c}{Time} & \multicolumn{1}{c}{$\sigma$} & \multicolumn{1}{|c}{\#} & \multicolumn{1}{c}{Time} & \multicolumn{1}{c}{$\sigma$} & \multicolumn{1}{|c}{\#} & \multicolumn{1}{c}{Time} & \multicolumn{1}{c}{$\sigma$} \\
        
	\hline
	tla-consensus $\hfill$  &  \cmark & 0 & 0 &  \cmark & 5 & 0 &  \cmark & \textbf{0} & 0 &  \cmark & 1 & 0 \\
	tla-tcommit $\hfill$  &  \cmark & \textbf{1} & 0 &  \xmark &  & &  \cmark & 1 & 0 &  \cmark & 2 & 0 \\
	i4-lock-server $\hfill$  &  \cmark & \textbf{1} & 0 &  \cmark & 2 & 0 &  \cmark & 1 & 0 &  \cmark & 1 & 0 \\
	ex-quorum-leader-election $\hfill$  &  \cmark & \textbf{3} & 0 &  \cmark & 32 & 0 &  \cmark & 10 & 1 &  \cmark & 21 & 3 \\
	pyv-toy-consensus-forall $\hfill$  &  \cmark & \textbf{3} & 1 &  \xmark &  & &  \cmark & 6 & 1 &  \cmark & 11 & 1 \\
	tla-simple $\hfill$  &  \cmark & 34 & 93 &  \cmark & 5 & 0 &  \xmark &  & & 2 & \textbf{3} & 0 \\
	ex-lockserv-automaton $\hfill$  &  \cmark & 9 & 3 &  \cmark & \textbf{3} & 0 &  \cmark & 21 & 1 &  \cmark & 11 & 0 \\
	tla-simpleregular $\hfill$  &  \cmark & \textbf{8} & 4 &  \xmark &  & &  \xmark &  & &  \cmark & 79 & 22 \\
	pyv-sharded-kv $\hfill$  &  \cmark & 8 & 1 &  \cmark & \textbf{4} & 0 &  \cmark & 6 & 0 &  \cmark & 22 & 0 \\
	pyv-lockserv $\hfill$  &  \cmark & 11 & 4 &  \cmark & \textbf{3} & 0 &  \cmark & 15 & 2 &  \cmark & 8 & 0 \\
	tla-twophase $\hfill$  &  \cmark & \textbf{15} & 3 &  \xmark &  & &  \cmark & 99 & 12 &  \cmark & 16 & 8 \\
	i4-learning-switch $\hfill$  &  \cmark & \textbf{20} & 8 &  \cmark & 22 & 0 &  \xmark &  & &  \xmark &  & \\
	ex-simple-decentralized-lock $\hfill$  &  \cmark & 20 & 0 &  \cmark & 14 & 0 &  \cmark & \textbf{4} & 0 &  \cmark & 4 & 0 \\
	i4-two-phase-commit $\hfill$  &  \cmark & 79 & 167 &  \cmark & \textbf{4} & 0 &  \cmark & 19 & 3 &  \cmark & 9 & 0 \\
	pyv-consensus-wo-decide $\hfill$  &  \cmark & \textbf{40} & 9 &  \cmark & 1226 & 37 &  \cmark & 107 & 16 &  \cmark & 82 & 45 \\
	pyv-consensus-forall $\hfill$  &  \cmark & \textbf{135} & 72 &  \cmark & 1042 & 36 &  \cmark & 398 & 86 &  \cmark & 2277 & 553 \\
	pyv-learning-switch $\hfill$  &  \cmark & \textbf{161} & 66 &  \cmark & 387 & 17 &  \cmark & 209 & 56 & 1 & 311 & 0 \\
	i4-chord-ring-maintenance $\hfill$  & 8 & \textbf{1289} & 1191 &  \xmark &  & &  \xmark &  & &  \xmark &  & \\
	pyv-sharded-kv-no-lost-keys $\hfill$  &  \cmark & \textbf{2} & 0 &  \xmark &  & &  \xmark &  & &  \cmark & 5 & 1 \\
	ex-naive-consensus $\hfill$  &  \cmark & \textbf{5} & 1 &  \xmark &  & &  \xmark &  & &  \cmark & 80 & 17 \\
	pyv-client-server-ae $\hfill$  &  \cmark & \textbf{1} & 0 &  \xmark &  & &  \xmark &  & &  \cmark & 630 & 130 \\
	ex-simple-election $\hfill$  &  \cmark & 172 & 522 &  \xmark &  & &  \xmark &  & &  \cmark & \textbf{38} & 8 \\
	pyv-toy-consensus-epr $\hfill$  &  \cmark & \textbf{14} & 8 &  \xmark &  & &  \xmark &  & &  \cmark & 47 & 12 \\
	ex-toy-consensus $\hfill$  &  \cmark & \textbf{11} & 5 &  \xmark &  & &  \xmark &  & &  \cmark & 22 & 4 \\
	pyv-client-server-db-ae $\hfill$  &  \cmark & \textbf{32} & 30 &  \xmark &  & &  \xmark &  & &  \xmark &  & \\
	pyv-hybrid-reliable-broadcast $\hfill$  & 6 & \textbf{157} & 211 &  \xmark &  & &  \xmark &  & & 6 & 2264 & 740 \\
	pyv-firewall $\hfill$  &  \cmark & \textbf{2} & 0 &  \xmark &  & &  \xmark &  & &  \cmark & 6 & 1 \\
	ex-majorityset-leader-election $\hfill$  &  \cmark & \textbf{63} & 47 &  \xmark &  & &  \xmark &  & &  \xmark &  & \\
	pyv-consensus-epr $\hfill$  & 2 & 1968 & 943 &  \xmark &  & &  \xmark &  & & 5 & \textbf{768} & 404 \\
	\hline
	\multicolumn{1}{l}{No. of problems solved \hfill (out of 29)} & \multicolumn{3}{|c}{\textbf{29}} & \multicolumn{3}{|c}{13} & \multicolumn{3}{|c}{14} & \multicolumn{3}{|c}{25}\\
	\multicolumn{1}{l}{Uniquely solved} & \multicolumn{3}{|c}{\textbf{3}} & \multicolumn{3}{|c}{0} & \multicolumn{3}{|c}{0} & \multicolumn{3}{|c}{0}\\
	\hline
	\multicolumn{1}{l}{For $11$ cases solved by all: \hfill $\sum$ Time } & \multicolumn{3}{|c}{\textbf{470}} & \multicolumn{3}{|c}{2727} & \multicolumn{3}{|c}{795} & \multicolumn{3}{|c}{2752}\\
	\hline
\end{tabular}
}
\captionsetup{justification=centering}
\caption{Statistical comparison of \icpo~against other state-of-the-art verifiers \\
\#: number of runs where successfully solved (out of 10) (\cmark~means 10, \xmark~means 0), Time: runtime mean (in seconds), $\sigma$: runtime standard deviation (in seconds)
}
\label{tab:table1_old}
\end{center}
\vspace{-20pt}
\end{table}

\clearpage
\onecolumn
\section{Comparison against Human-Written Invariants}
\label{sec:eval_human}
\noindent
Figure~\ref{plot:inv} compares \icpo's automatically-generated inductive invariants against the human-written proofs on several metrics. Our evaluation shows \icpo~produces compact proofs of sizes comparable to the manually-written inductive invariants, even shorter than the human proofs on several occasions. As a side benefit, \icpo's inductive invariants are pretty-printed in the Ivy format~\cite{ivy}, and thus, can also be independently checked/validated through Ivy.


\begin{figure}[H]
\centering
    \subfloat{
        \includegraphics[scale=0.35]{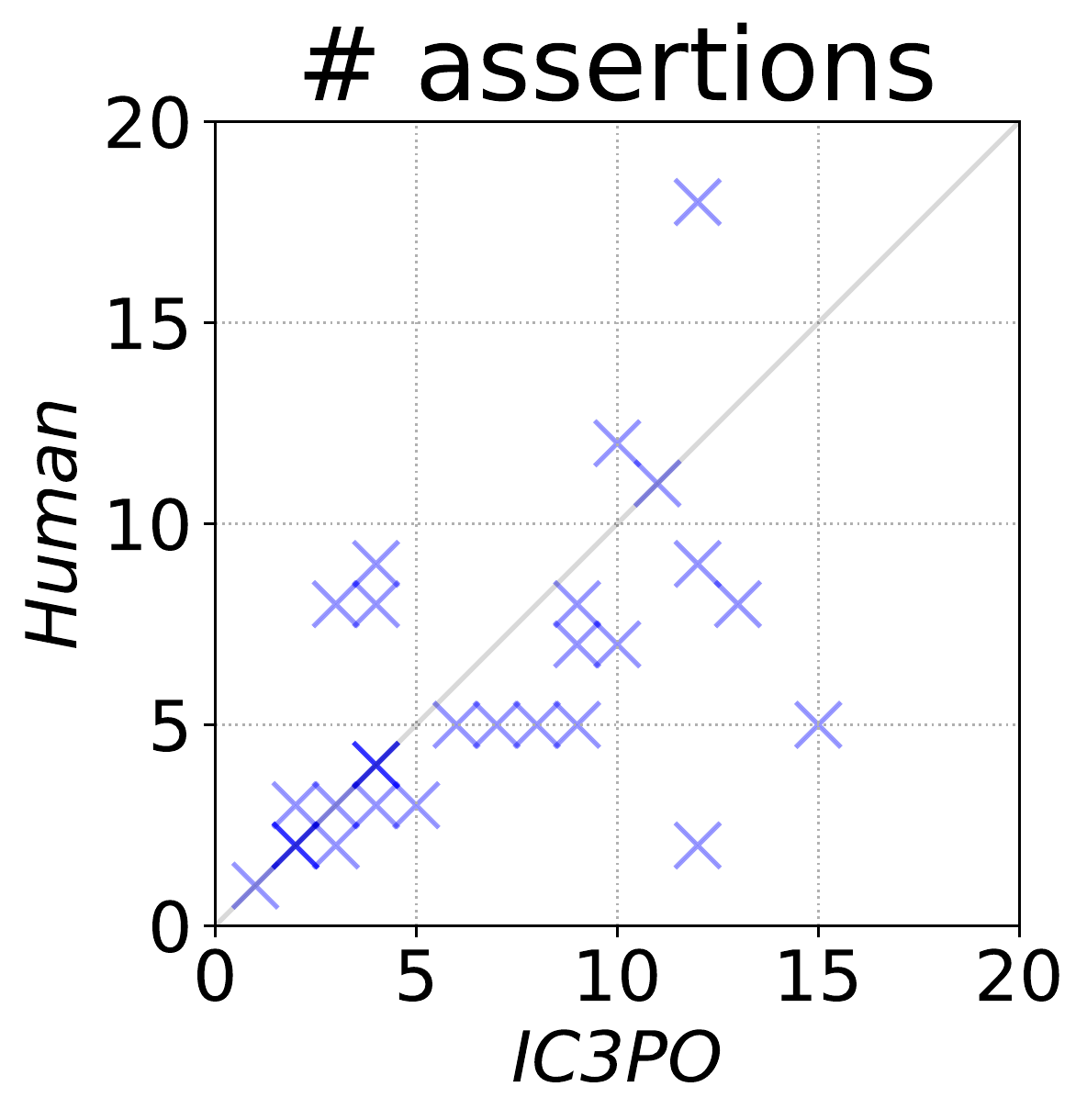}
    }
    \hfill
    \subfloat{
        \includegraphics[scale=0.35]{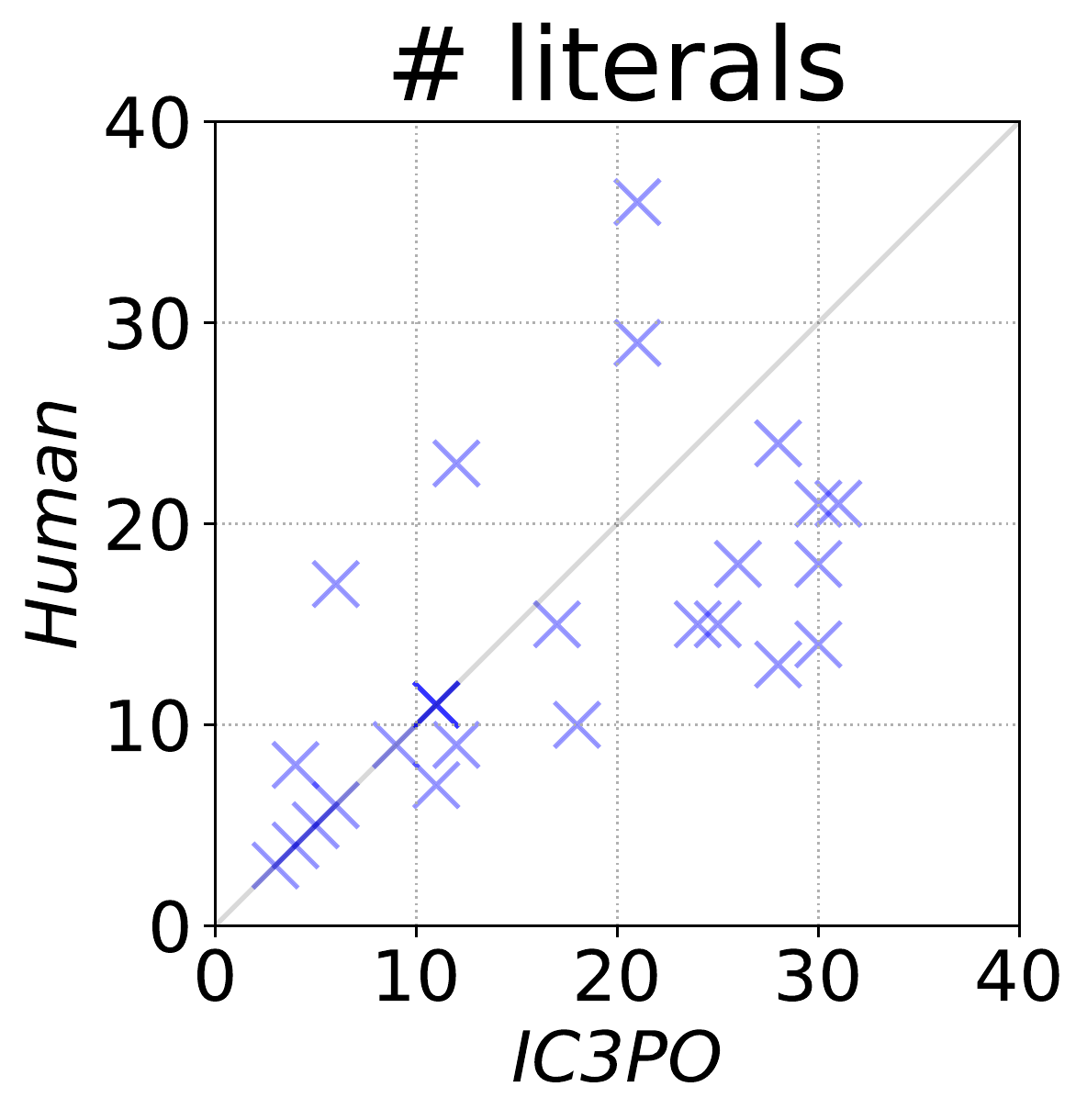}
    }
    \\
    \vspace{10pt}
    \subfloat{
        \includegraphics[scale=0.35]{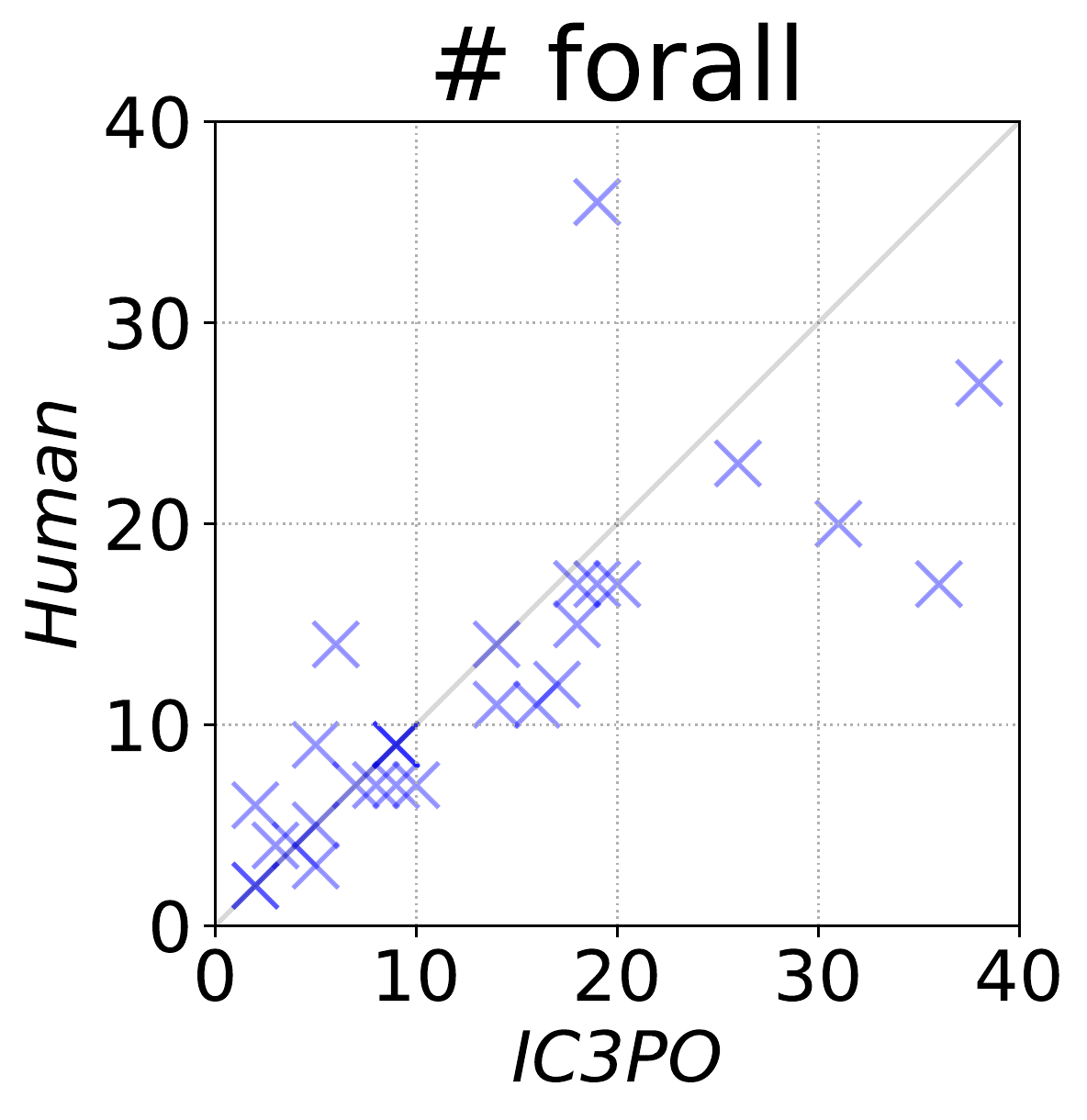}
    }
    \hfill
    \subfloat{
        \includegraphics[scale=0.35]{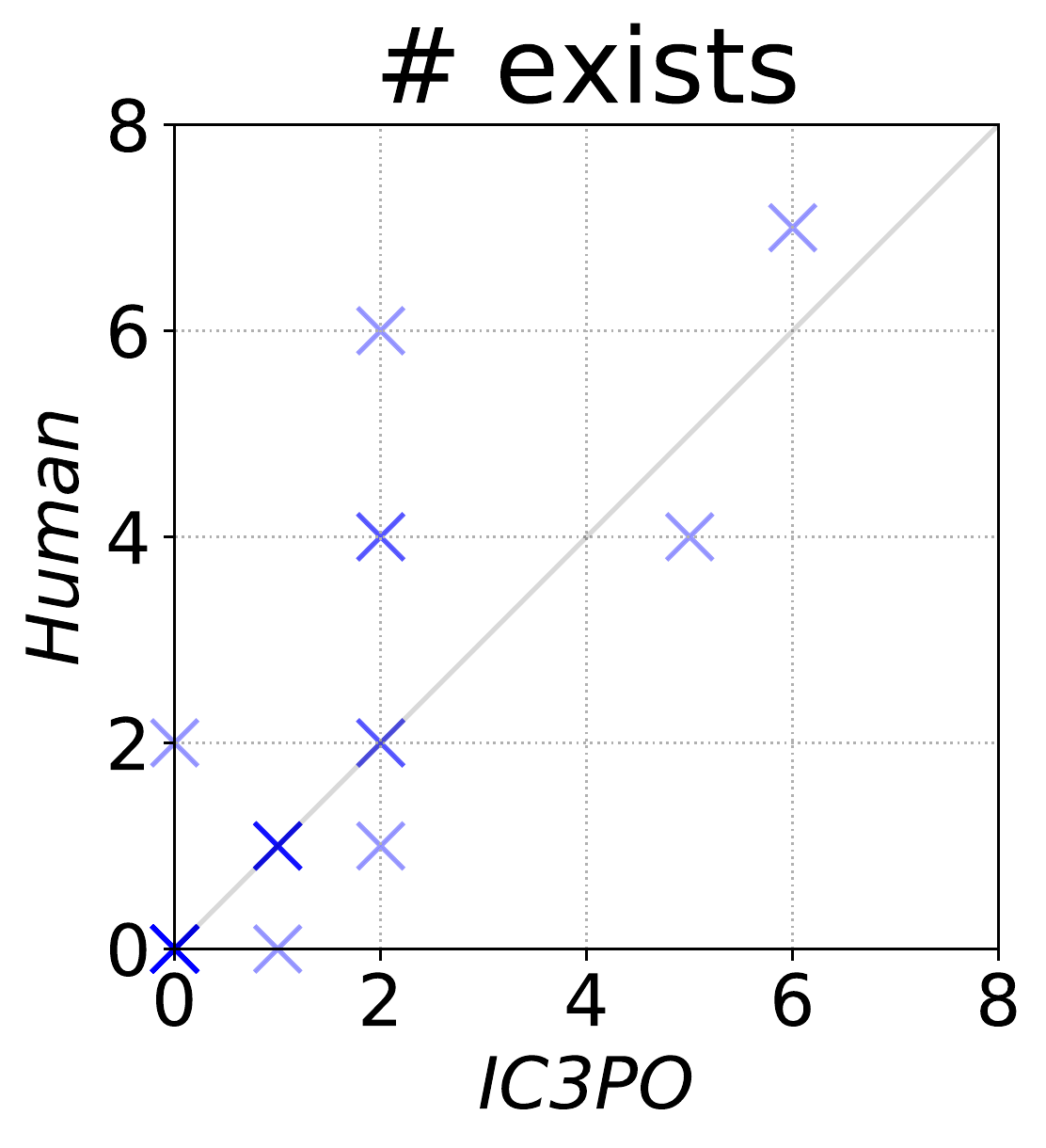}
    }
\captionsetup{justification=centering}
\caption{Comparison of \icpo's inductive invariant against \textit{human-written} proof \\
\icpo~is on x-axis, \textit{human-written} on y-axis
}
\label{plot:inv}
\end{figure}

\clearpage
\onecolumn
\section{Ordered Domains, Ring Topology and Special Variables}
\label{sec:eval_order}
\noindent
This section describes an extension to \icpo that allows handling totally-ordered domains, as well as further details relating to ring topology and special variables (along with a preliminary evaluation).

\begin{table}[H]
\small
\setlength\tabcolsep{3pt}
\begin{center}
\resizebox{\textwidth}{!}{
\begin{tabular}{l|r|rrr|rrr|rrr|rrr}
    \multicolumn{1}{c}{} & \multicolumn{1}{c}{\makebox[20pt][c]{\textit{Human}}} & \multicolumn{3}{c}{\icpo} & \multicolumn{3}{c}{\ifour} & \multicolumn{3}{c}{\updr} & \multicolumn{3}{c}{\folic} \\
    \hline
    \multicolumn{1}{c}{Protocol (\#13)} & \multicolumn{1}{|c}{Inv} & \multicolumn{1}{|c}{Time} & \multicolumn{1}{c}{Inv} & \multicolumn{1}{c}{SMT} & \multicolumn{1}{|c}{Time} & \multicolumn{1}{c}{Inv} & \multicolumn{1}{c}{SMT} & \multicolumn{1}{|c}{Time} & \multicolumn{1}{c}{Inv} & \multicolumn{1}{c}{SMT} & \multicolumn{1}{|c}{Time} & \multicolumn{1}{c}{Inv} & \multicolumn{1}{c}{SMT} \\
    
	\hline
	ex-distributed-lock-abstract $\hfill$  $<$  & 12 & \textbf{15} & 11 & 946 &  \multicolumn{2}{l}{ \textcolor{failcolor}{timeout} } & &  \multicolumn{2}{l}{ \textcolor{failcolor}{timeout} } & &  \multicolumn{2}{l}{ \textcolor{failcolor}{timeout} } & \\
	ex-decentralized-lock $\hfill$  $<$  & 4 & \textbf{25} & 5 & 654 & 288 & 32 & 104616 &  \multicolumn{2}{l}{ \textcolor{failcolor}{timeout} } & &  \multicolumn{2}{l}{ \textcolor{failcolor}{timeout} } & \\
	ex-distributed-lock-maxheld $\hfill$  $<$  & 6 & \textbf{58} & 10 & 1866 & 422 & 73 & 100749 &  \multicolumn{2}{l}{ \textcolor{failcolor}{timeout} } & & 3210 & 48 & 4557 \\
	pyv-ticket $\hfill$  $<$  & 14 & \textbf{65} & 8 & 1896 &  \multicolumn{2}{l}{ \textcolor{failcolor}{error} } & & 228 & 13 & 15936 & 98 & 26 & 3177 \\
	i4-database-chain-replication $\hfill$  $<$  & 9 & 98 & 6 & 1382 & \textbf{20} & 10 & 6111 &  \multicolumn{2}{l}{ \textcolor{failcolor}{timeout} } & & 1222 & 16 & 5455 \\
	ex-decentralized-lock-abstract $\hfill$  $<$  & 6 & \textbf{126} & 18 & 5069 &  \multicolumn{2}{l}{ \textcolor{failcolor}{error} } & &  \multicolumn{2}{l}{ \textcolor{failcolor}{timeout} } & &  \multicolumn{2}{l}{ \textcolor{failcolor}{timeout} } & \\
	i4-distributed-lock $\hfill$  $<$  & 7 & \textbf{155} & 10 & 3472 & 3280 & 102 & 410364 &  \multicolumn{2}{l}{ \textcolor{failcolor}{timeout} } & & 1191 & 64 & 4875 \\
	ex-ring-not-dead $\hfill$ $\circlearrowright$ $<$  & 2 & 10 & 2 & 161 &  \multicolumn{2}{l}{ \textcolor{failcolor}{unknown} } & 3327 &  \multicolumn{2}{l}{ \textcolor{failcolor}{unknown} } & 28 & \textbf{6} & 3 & 100 \\
	ex-ring $\hfill$  $\circlearrowright$ $<$  & 3 & 11 & 3 & 269 & \textbf{6} & 9 & 678 & 9 & 2 & 662 & 7 & 3 & 248 \\
	ex-ring-id-not-dead-limited $\hfill$ $\circlearrowright$ $<$  & 2 & 24 & 2 & 250 &  \multicolumn{2}{l}{ \textcolor{failcolor}{unknown} } & 29083 &  \multicolumn{2}{l}{ \textcolor{failcolor}{unknown} } & 31 & \textbf{7} & 3 & 81 \\
	pyv-ring-id-not-dead $\hfill$ $\circlearrowright$ $<$  & 2 & 37 & 2 & 275 &  \multicolumn{2}{l}{ \textcolor{failcolor}{unknown} } & 182325 &  \multicolumn{2}{l}{ \textcolor{failcolor}{unknown} } & 31 & \textbf{8} & 3 & 86 \\
	pyv-ring-id $\hfill$  $\circlearrowright$ $<$  & 4 & 73 & 4 & 869 & 420 & 11 & 225789 & 99 & 3 & 4107 & \textbf{28} & 9 & 594 \\
	i4-leader-election-in-ring $\hfill$  $\circlearrowright$ $<$  & 6 & 323 & 5 & 2907 & 749 & 25 & 359776 & 114 & 3 & 4229 & \textbf{59} & 17 & 1378 \\
	\hline
	\multicolumn{2}{l}{No. of problems solved \hfill (out of 13)} & \multicolumn{3}{|c}{\textbf{13}} & \multicolumn{3}{|c}{7} & \multicolumn{3}{|c}{4} & \multicolumn{3}{|c}{10}\\
	\multicolumn{2}{l}{Uniquely solved} & \multicolumn{3}{|c}{\textbf{2}} & \multicolumn{3}{|c}{0} & \multicolumn{3}{|c}{0} & \multicolumn{3}{|c}{0}\\
	\hline
	\multicolumn{2}{l}{For $3$ cases solved by all: \hfill $\sum$ Time } & \multicolumn{3}{|c}{407} & \multicolumn{3}{|c}{1176} & \multicolumn{3}{|c}{224} & \multicolumn{3}{|c}{\textbf{95}}\\
	\multicolumn{2}{l}{\hfill $\sum$ Inv \hskip 5pt ~} & \multicolumn{3}{|c}{12} & \multicolumn{3}{|c}{45} & \multicolumn{3}{|c}{\textbf{8}} & \multicolumn{3}{|c}{29}\\
	\multicolumn{2}{l}{\hfill $\sum$ SMT } & \multicolumn{3}{|c}{4045} & \multicolumn{3}{|c}{586243} & \multicolumn{3}{|c}{8998} & \multicolumn{3}{|c}{\textbf{2220}}\\
	\hline
\end{tabular}
}
\captionsetup{justification=centering}
\caption{Comparison of \icpo~against other state-of-the-art verifiers \\
Time: runtime in seconds,
Inv: \# assertions in the inductive invariant, \\
SMT: \# SMT solver queries made,
$\circlearrowright$ indicates protocol has a ring topology, $<$ indicates protocol has a totally-ordered domain
}
\label{tab:table2}
\end{center}
\vspace{-20pt}
\end{table}

Ordered domains like \textit{epoch}, \textit{time}, etc. are not symmetric, which makes such domains unsuitable to directly apply a symmetry argument. Specifically, restricting an unbounded ordered domain to a finite size results in introducing boundary cases with a ``max'' element, complicating finite-domain behavior.

Even in the presence of ordered domains, symmetry-aware learning can still be applied to all the un-ordered domains while leaving the ordered domains as unbounded. As an initial exploration, we devised a hybrid procedure in \icpo~where ordered domains are handled in an unbounded fashion, in the same manner as in \updr, while all other domains are handled in the \symic-style symmetry-aware and finite manner. We use \updr's \textit{diagram-based abstraction} to infer quantifiers for the ordered domain, while using \symboost~(Algorithm~\ref{alg:symgen}) for the un-ordered domains.\footnote{We refer the reader to~\cite{10.1145/3022187} for a complete description of incremental induction with diagram-based abstraction.}

For the protocols that involve a ring topology, a ring domain, generally composed of identical components arranged in a ring topology, retains domain symmetry since the position of each individual component in the ring is left uninitialized and can be arbitrarily permuted. Hence, \symic~can be directly applied. The same is true for protocols that have special components, like a special $start\_node$ that initially holds the lock in a distributed lock. Non-Boolean functions and variables are modeled in relational form with equality predicates. For example, permuting the predicate $(start\_node = \fin{n_1})$ with the permutation $(\fin{n_1~n_2})$ gives the permuted predicate $(start\_node = \fin{n_2})$. \icpo~exploits the symmetry in the sort domains, not symmetries over the protocol symbols (i.e., relations, functions and variables), and hence is unaffected by the presence of special protocol symbols.

Table~\ref{tab:table2} summarizes the experimental results for $13$ protocols with totally-ordered domains, collected again from~\cite{ma2019i4,pldi20folic3,ivybench}. \icpo~solves all $13$ problems and shows the advantages of symmetry-aware learning even when applied only to a subset of protocol's domains. 
We believe additional exploration is needed for these cases, where the non-symmetric regularity in totally-ordered domains can be further utilized to improve learning during incremental induction.

\clearpage
\onecolumn
\section{Finite Instance Sizes used in Experiments}
\label{app:sizes}
Table~\ref{tab:finite_sizes_ic3po} lists down the initial base instance sizes used for \icpo runs in the evaluation (Section~\ref{sec:evaluation}) for each protocol. The table also includes the final $cutoff$ instance sizes reached, where the corresponding $Inv$ generalizes/saturates to be an inductive proof for any size. Note again that \icpo updates the instance sizes automatically, as described in Section~\ref{sec:invcheck}.

\begin{table}[!b]
\small
\setlength\tabcolsep{3pt}
\begin{center}
\resizebox{0.98\textwidth}{!}{
\begin{tabular}{l|l}
    \hline
    \multicolumn{1}{c|}{Protocol} & \multicolumn{1}{c}{Finite instance sizes used for \icpo} \\
	\hline
	tla-consensus $\hfill$  & $ \fin{value} = 2$\\
	tla-tcommit $\hfill$  & $ \fin{resource\text{-}manager} = 2$\\
	i4-lock-server $\hfill$  & $ \fin{client} = 2,~\fin{server} = 1$\\
	ex-quorum-leader-election $\hfill$ $E$  & $ \fin{node} = 2 \mapsto 3,~\fin{nset} = 2$\\
	pyv-toy-consensus-forall $\hfill$ $E$  & $ \fin{node} = 2 \mapsto 3,~\fin{quorum} = 1 \mapsto 3,~\fin{value} = 2$\\
	tla-simple $\hfill$ $\circlearrowright$ $E$  & $ \fin{node} = 2,~\fin{pcstate} = 3,~\fin{value} = 2 \mapsto 3$\\
	ex-lockserv-automaton $\hfill$  & $ \fin{node} = 2$\\
	tla-simpleregular $\hfill$ $\circlearrowright$ $E$  & $ \fin{node} = 2,~\fin{pcstate} = 4,~\fin{value} = 2 \mapsto 3$\\
	pyv-sharded-kv $\hfill$  & $ \fin{key} = 2,~\fin{node} = 2,~\fin{value} = 2$\\
	pyv-lockserv $\hfill$  & $ \fin{node} = 2$\\
	tla-twophase $\hfill$  & $ \fin{resource\text{-}manager} = 2$\\
	i4-learning-switch $\hfill$  & $ \fin{node} = 2 \mapsto 3,~\fin{packet} = 1$\\
	ex-simple-decentralized-lock $\hfill$  & $ \fin{node} = 2 \mapsto 4$\\
	i4-two-phase-commit $\hfill$  & $ \fin{node} = 4$\\
	pyv-consensus-wo-decide $\hfill$ $E$  & $ \fin{node} = 2 \mapsto 3,~\fin{quorum} = 1 \mapsto 3$\\
	pyv-consensus-forall $\hfill$ $E$  & $ \fin{node} = 2 \mapsto 3,~\fin{quorum} = 1 \mapsto 3,~\fin{value} = 2$\\
	pyv-learning-switch $\hfill$ $E$  & $ \fin{node} = 2 \mapsto 4$\\
	i4-chord-ring-maintenance $\hfill$ $\circlearrowright$ $E$  & $ \fin{node} = 3 \mapsto 5$\\
	pyv-sharded-kv-no-lost-keys $\hfill$ $E$  & $ \fin{key} = 2,~\fin{node} = 2,~\fin{value} = 2$\\
	ex-naive-consensus $\hfill$ $E$  & $ \fin{node} = 3,~\fin{quorum} = 3,~\fin{value} = 3$\\
	pyv-client-server-ae $\hfill$ $E$  & $ \fin{node} = 2,~\fin{request} = 2 \mapsto 3,~\fin{response} = 2$\\
	ex-simple-election $\hfill$ $E$  & $ \fin{acceptor} = 2 \mapsto 3,~\fin{proposer} = 2,~\fin{quorum} = 1 \mapsto 3$\\
	pyv-toy-consensus-epr $\hfill$ $E$  & $ \fin{node} = 2 \mapsto 3,~\fin{quorum} = 1 \mapsto 3,~\fin{value} = 2$\\
	ex-toy-consensus $\hfill$ $E$  & $ \fin{node} = 2 \mapsto 3,~\fin{quorum} = 1 \mapsto 3,~\fin{value} = 2$\\
	pyv-client-server-db-ae $\hfill$ $E$  & $ \fin{db\text{-}request\text{-}id} = 2 \mapsto 3,~\fin{node} = 2,~\fin{request} = 2 \mapsto 3,~\fin{response} = 2$\\
	pyv-hybrid-reliable-broadcast $\hfill$ $E$  & $ \fin{node} = 2 \mapsto 3,~\fin{quorum\text{-}a} = 2 \mapsto 3,~\fin{quorum\text{-}b} = 2$\\
	pyv-firewall $\hfill$ $E$  & $ \fin{node} = 2 \mapsto 3$\\
	ex-majorityset-leader-election $\hfill$ $E$  & $ \fin{node} = 2 \mapsto 3,~\fin{nodeset} = 2 \mapsto 3$\\
	pyv-consensus-epr $\hfill$ $E$  & $ \fin{node} = 2 \mapsto 3,~\fin{quorum} = 1 \mapsto 3,~\fin{value} = 2$\\
	\hline
	ex-distributed-lock-abstract $\hfill$ $<$  & $ \fin{epoch} = \infty,~\fin{node} = 2$\\
	ex-decentralized-lock $\hfill$ $<$  & $ \fin{node} = 2,~\fin{time} = \infty$\\
	ex-distributed-lock-maxheld $\hfill$ $<$  & $ \fin{epoch} = \infty,~\fin{node} = 2$\\
	pyv-ticket $\hfill$ $<$  & $ \fin{thread} = 2 \mapsto 3,~\fin{ticket} = \infty$\\
	i4-database-chain-replication $\hfill$ $E$ $<$  & $ \fin{key} = 1,~\fin{node} = 2,~\fin{operation} = 2 \mapsto 3,~\fin{transaction} = \infty$\\
	ex-decentralized-lock-abstract $\hfill$ $<$  & $ \fin{node} = 2 \mapsto 4,~\fin{time} = \infty$\\
	i4-distributed-lock $\hfill$ $<$  & $ \fin{epoch} = \infty,~\fin{node} = 2$\\
	ex-ring-not-dead $\hfill$ $\circlearrowright$ $E$ $<$  & $ \fin{node} = 3$\\
	ex-ring $\hfill$ $\circlearrowright$ $<$  & $ \fin{node} = 3$\\
	ex-ring-id-not-dead-limited $\hfill$ $\circlearrowright$ $E$ $<$  & $ \fin{id} = 3,~\fin{node} = 3$\\
	pyv-ring-id-not-dead $\hfill$ $\circlearrowright$ $E$ $<$  & $ \fin{id} = \infty,~\fin{node} = 3$\\
	pyv-ring-id $\hfill$ $\circlearrowright$ $<$  & $ \fin{id} = \infty,~\fin{node} = 3$\\
	i4-leader-election-in-ring $\hfill$ $\circlearrowright$ $<$  & $ \fin{id} = \infty,~\fin{node} = 3$\\
	\hline
\end{tabular}
}
\captionsetup{justification=centering, belowskip=0pt}
\caption{
Finite instance sizes used for \icpo
}
\captionsetup{justification=raggedright, aboveskip=0pt}
\caption*{
   $\s = x$ denotes sort $\s$ has both initial base size and final cutoff size $x$
\\ $\s = x \mapsto y$ denotes sort $\s$ has initial size $x$ and final cutoff size $y$ (incrementally increased by \icpo~automatically)
\\ $\s = \infty$ denote the totally-ordered sort $\s$ is left unbounded
\\ $\circlearrowright$ indicates protocol has a ring topology, $<$ indicates protocol has an ordered domain \\
$E$ indicates the protocol description has $\exists$
}
\label{tab:finite_sizes_ic3po}
\end{center}
\end{table}

Table~\ref{tab:finite_sizes_i4} lists down the instance sizes used for \ifour runs in the evaluation (Section~\ref{sec:evaluation}) for each protocol.

\begin{table}[!b]
\small
\setlength\tabcolsep{3pt}
\begin{center}
\resizebox{\textwidth}{!}{
\begin{tabular}{l|l}
    \hline
    \multicolumn{1}{c|}{Protocol} & \multicolumn{1}{c}{Finite instance sizes used for \ifour} \\
	\hline
	tla-consensus $\hfill$  & $ \fin{value} = 2$\\
	tla-tcommit $\hfill$  & $ \fin{resource\text{-}manager} = 2$\\
	i4-lock-server $\hfill$  & $ \fin{client} = 2,~\fin{server} = 1$\\
	ex-quorum-leader-election $\hfill$ $E$  & $ \fin{node} = 3,~\fin{nset} = 3$\\
	pyv-toy-consensus-forall $\hfill$ $E$  & $ \fin{node} = 3,~\fin{quorum} = 3,~\fin{value} = 2$\\
	tla-simple $\hfill$ $\circlearrowright$ $E$  & $ \fin{node} = 3,~\fin{pcstate} = 3,~\fin{value} = 3$\\
	ex-lockserv-automaton $\hfill$  & $ \fin{node} = 2$\\
	tla-simpleregular $\hfill$ $\circlearrowright$ $E$  & $ \fin{node} = 3,~\fin{pcstate} = 4,~\fin{value} = 3$\\
	pyv-sharded-kv $\hfill$  & $ \fin{key} = 2,~\fin{node} = 2,~\fin{value} = 2$\\
	pyv-lockserv $\hfill$  & $ \fin{node} = 2$\\
	tla-twophase $\hfill$  & $ \fin{resource\text{-}manager} = 3$\\
	i4-learning-switch $\hfill$  & $ \fin{node} = 3,~\fin{packet} = 2$\\
	ex-simple-decentralized-lock $\hfill$  & $ \fin{node} = 4$\\
	i4-two-phase-commit $\hfill$  & $ \fin{node} = 5$\\
	pyv-consensus-wo-decide $\hfill$ $E$  & $ \fin{node} = 3,~\fin{quorum} = 3$\\
	pyv-consensus-forall $\hfill$ $E$  & $ \fin{node} = 3,~\fin{quorum} = 3,~\fin{value} = 2$\\
	pyv-learning-switch $\hfill$ $E$  & $ \fin{node} = 4$\\
	i4-chord-ring-maintenance $\hfill$ $\circlearrowright$ $E$  & $ \fin{node} = 4$\\
	pyv-sharded-kv-no-lost-keys $\hfill$ $E$  & $ \fin{key} = 3,~\fin{node} = 3,~\fin{value} = 3$\\
	ex-naive-consensus $\hfill$ $E$  & $ \fin{node} = 3,~\fin{quorum} = 3,~\fin{value} = 3$\\
	pyv-client-server-ae $\hfill$ $E$  & $ \fin{node} = 3,~\fin{request} = 3,~\fin{response} = 3$\\
	ex-simple-election $\hfill$ $E$  & $ \fin{acceptor} = 3,~\fin{proposer} = 2,~\fin{quorum} = 3$\\
	pyv-toy-consensus-epr $\hfill$ $E$  & $ \fin{node} = 3,~\fin{quorum} = 3,~\fin{value} = 2$\\
	ex-toy-consensus $\hfill$ $E$  & $ \fin{node} = 3,~\fin{quorum} = 3,~\fin{value} = 2$\\
	pyv-client-server-db-ae $\hfill$ $E$  & $ \fin{db\text{-}request\text{-}id} = 3,~\fin{node} = 3,~\fin{request} = 3,~\fin{response} = 3$\\
	pyv-hybrid-reliable-broadcast $\hfill$ $E$  & $ \fin{node} = 3,~\fin{quorum\text{-}a} = 3,~\fin{quorum\text{-}b} = 3$\\
	pyv-firewall $\hfill$ $E$  & $ \fin{node} = 3$\\
	ex-majorityset-leader-election $\hfill$ $E$  & $ \fin{node} = 3,~\fin{nodeset} = 3$\\
	pyv-consensus-epr $\hfill$ $E$  & $ \fin{node} = 3,~\fin{quorum} = 3,~\fin{value} = 2$\\
	\hline
	ex-distributed-lock-abstract $\hfill$ $<$  & $ \fin{epoch} = 4,~\fin{node} = 2$\\
	ex-decentralized-lock $\hfill$ $<$  & $ \fin{node} = 2,~\fin{time} = 4$\\
	ex-distributed-lock-maxheld $\hfill$ $<$  & $ \fin{epoch} = 4,~\fin{node} = 2$\\
	pyv-ticket $\hfill$ $<$  & $ \fin{thread} = 3,~\fin{ticket} = 5$\\
	i4-database-chain-replication $\hfill$ $E$ $<$  & $ \fin{key} = 1,~\fin{node} = 2,~\fin{operation} = 3,~\fin{transaction} = 3$\\
	ex-decentralized-lock-abstract $\hfill$ $<$  & $ \fin{node} = 4,~\fin{time} = 4$\\
	i4-distributed-lock $\hfill$ $<$  & $ \fin{epoch} = 4,~\fin{node} = 2$\\
	ex-ring-not-dead $\hfill$ $\circlearrowright$ $E$ $<$  & $ \fin{node} = 3$\\
	ex-ring $\hfill$ $\circlearrowright$ $<$  & $ \fin{node} = 3$\\
	ex-ring-id-not-dead-limited $\hfill$ $\circlearrowright$ $E$ $<$  & $ \fin{id} = 3,~\fin{node} = 3$\\
	pyv-ring-id-not-dead $\hfill$ $\circlearrowright$ $E$ $<$  & $ \fin{id} = 4,~\fin{node} = 3$\\
	pyv-ring-id $\hfill$ $\circlearrowright$ $<$  & $ \fin{id} = 4,~\fin{node} = 3$\\
	i4-leader-election-in-ring $\hfill$ $\circlearrowright$ $<$  & $ \fin{id} = 4,~\fin{node} = 3$\\
	\hline
\end{tabular}
}
\captionsetup{justification=centering}
\caption{
Finite instance sizes used for \ifour
\\ $\circlearrowright$ indicates protocol has a ring topology, $<$ indicates protocol has an ordered domain \\
$E$ indicates the protocol description has $\exists$
}
\label{tab:finite_sizes_i4}
\end{center}
\end{table}

\end{subappendices}

\end{document}